%% file: main.tex
\documentclass[letterpaper,twocolumn,10pt]{article}
\usepackage{usenix-2020-09}

\usepackage[utf8]{inputenc} 
\usepackage[T1]{fontenc}    
\usepackage{hyperref}       
\hypersetup{
    colorlinks,
    linkcolor={blue!80!black},
    citecolor={blue!80!black},
    urlcolor={blue!30!black}
}
\usepackage{url}            
\usepackage{nicefrac}       
\usepackage{microtype}      
\usepackage{amsmath}
\usepackage{amsthm}
\usepackage{enumitem}
\usepackage{color}
\usepackage{thmtools,thm-restate}	
\usepackage{outlines}
\usepackage{tikz}
\usetikzlibrary{arrows,positioning,shapes.geometric,automata}
\usepackage{multirow}
\usepackage{caption}
\usepackage{listings}
\input{MacrosMarkup}
\Bleck		
\MakeMarkups[Kostas]{K}{\Colour{Brown}}
\MakeMarkups[George]{G}{\Colour{DarkGreen}}
\MakeMarkups[Thodoris]{T}{\Colour{PurplePlum}}

\newcommand{\Wide}[1]{~#1~}

\newcommand\eatspace[1]{}
\def\sectionautorefname{\S\kern.03em\eatspace}  
\def\subsectionautorefname{\S\kern.03em\eatspace}
\def\subsubsectionautorefname{\S\kern.03em\eatspace}

\begin{document}

\twocolumn
\setcounter{section}{0} 
\setcounter{page}{1}

\title{Tor circuit fingerprinting defenses using adaptive padding}

\author{
{\rm George Kadianakis}\\
University of Athens \\
The Tor Project \\
\and
{\rm Theodoros Polyzos}\\
University of Athens \\
\and
{\rm Mike Perry}\\
The Tor Project \\
\and
{\rm Kostas Chatzikokolakis}\\
University of Athens \\
} 

\maketitle

\begin{abstract}

  Online anonymity and privacy has been based on confusing the adversary by
  creating indistinguishable network elements. Tor is the largest and most widely
  deployed anonymity system, designed against realistic modern
  adversaries. Recently, researchers have managed to fingerprint Tor's circuits
  -- and hence the type of underlying traffic -- simply by capturing and
  analyzing traffic traces. In this work, we study the \emph{circuit
    fingerprinting} problem, isolating it from website fingerprinting, and
  revisit previous findings in this model, showing that accurate attacks are
  possible even when the application-layer traffic is identical. We then
  proceed to incrementally create defenses against circuit fingerprinting,
  using a generic adaptive padding framework for Tor based on WTF-PAD. We
  present a simple defense which delays a fraction of the traffic,
  as well as a more advanced
  one which can effectively hide onion service circuits with zero delays.  We thoroughly evaluate both defenses, both analytically and
  experimentally, discovering new subtle fingerprints, but also showing the
  effectiveness of our defenses.

\end{abstract}

\input{1-introduction}
\input{2-model}
\input{3-methodology}
\input{4-vanilla}
\input{5-wtf-infrastructure}
\input{6-padding-strategies}
\input{7-high-latency}
\input{8-low-latency}
\input{9-discussion}
\input{10-conclusion}











\bibliography{biblio}
\bibliographystyle{plain}

\appendix

\input{appendix.tex}

\clearpage
\MakeEndNotes
\onecolumn
\MarkupsHowto 

\end{document}


%% file: MacrosMarkup.tex
\providecommand{\UsePackageFor}[2]{ \ifx#2\undefined\usepackage{#1}\fi }

\UsePackageFor{amssymb}{\boxtimes}
\usepackage[normalem]{ulem}		
\usepackage{bbding}				
\usepackage{endnotes}			
\usepackage{hyperref}			
\usepackage{hyperendnotes}	
\usepackage{pdfcolfoot}			

\ifdefined\OrigFootnote\relax\else
	\newenvironment{FootnoteContent}{}{}
	\let\OrigFootnote\footnote
	\let\OrigFootnoteText\footnotetext
	\renewcommand{\footnotetext}[1]{\OrigFootnoteText{\begin{FootnoteContent}#1\end{FootnoteContent}}}
	\renewcommand{\footnote    }[1]{\OrigFootnote    {\begin{FootnoteContent}#1\end{FootnoteContent}}}
\fi

\definecolor{PurplePlum}{rgb}{0.1,0,0.55} 
\definecolor{Brown}{rgb}{0.5,.25,0}
\definecolor{Orange}{rgb}{1,.3,0}
\definecolor{Gray}{rgb}{.7,.7,.7}
\definecolor{DarkGreen}{rgb}{.1,.41,.1}

\newif\ifBleck
\newcommand\Bleck {\Blecktrue} 

\newcommand\Colour[1] {\color{#1}}

\newcommand\PrintToCLinks{	
  {\Colour{blue}\mbox{
    \hyperlink{w1619}{\sf$\rightarrow$~top}\quad
    \hyperlink{w1031}{\sf$\rightarrow$~toc}\quad
    \hyperlink{w1148}{\sf$\rightarrow$~lof}\quad
    \hyperlink{GreenRoom}{\sf$\rightarrow$~gr}\quad
    \hyperlink{EndNotes}{\sf$\rightarrow$~en}\quad
    \hyperlink{Sargasso}{\sf$\rightarrow$~sg}\quad
    \hyperlink{Index}{\sf$\rightarrow$~idx}
  }}
}
\makeatletter 
\newcommand\ToCLinks{
  \ifx\@onlypreamble\@notprerr		
    \hypertarget{w1619}{}			
  \else
    \AtBeginDocument{\hypertarget{w1619}{}}	
  \fi

  \ifBleck\else	
    \ifdefined\cofoot
      \cofoot{\PrintToCLinks}
      \cefoot{\PrintToCLinks}
    \else
      \def\@oddfoot{\PrintToCLinks}
      \def\@evenfoot{\PrintToCLinks}
    \fi
 \fi
}
\makeatother

\newif\ifEndNotes 

\newcommand\FnSym{{\scriptsize\PencilLeftDown\kern.1em}}		
\newcommand\EnSym {{$\bigtriangledown$}}

\def\MarkupsHowto{} 
\newcommand{\MarkupsHowtoAdd}[1]{\expandafter\def\expandafter\MarkupsHowto\expandafter{\MarkupsHowto{}#1}} 

\newif\ifMarkupsHowtoPrinted 
\newif\ifSuppress 

\newcommand\MakeMarkups[3][.]{

     \Suppressfalse
     \ifBleck\Suppresstrue\fi
     \ifx0#1\Suppresstrue\fi
     \ifx1#1\Suppressfalse\fi
     
     \expandafter\providecommand\csname#2x\endcsname {} 
     \ifSuppress\expandafter\renewcommand\csname#2x\endcsname{\relax}\else
                       \expandafter\renewcommand\csname#2x\endcsname{#3}\fi
                       
     \expandafter\providecommand\csname#2\endcsname {} 
     \ifSuppress\expandafter\renewcommand\csname#2\endcsname[1]{##1}\else
                       \expandafter\renewcommand\csname#2\endcsname[1]{{\csname#2x\endcsname##1}}\fi

     \expandafter\providecommand\csname#2d\endcsname {} 
     \ifSuppress\expandafter\renewcommand\csname#2d\endcsname[1]{\relax}\else
                       \expandafter\renewcommand\csname#2d\endcsname[1]{{\csname#2x\endcsname\sout{##1}}}\fi
                       
     \expandafter\providecommand\csname#2r\endcsname {} 
     \ifSuppress\expandafter\renewcommand\csname#2r\endcsname[2]{{##2}}\else
                       \expandafter\renewcommand\csname#2r\endcsname[2]{\csname#2d\endcsname{##1} \csname#2\endcsname{##2}}\fi

     \expandafter\providecommand\csname#2i\endcsname {} 
     \ifSuppress\expandafter\renewcommand\csname#2i\endcsname[1]{\relax}\else
                       \expandafter\renewcommand\csname#2i\endcsname[1]{\csname#2\endcsname{##1}}\fi

     \expandafter\providecommand\csname#2t\endcsname {} 
     \ifSuppress\expandafter\renewcommand\csname#2t\endcsname[1]{\relax}\else
                       \expandafter\renewcommand\csname#2t\endcsname[1]{{\csname#2x\endcsname{\mbox{$\langle\!\langle$}##1{\csname#2x\endcsname\mbox{$\rangle\!\rangle$}}}}}\fi 

     \expandafter\providecommand\csname#2b\endcsname {} 
     \ifSuppress\expandafter\renewcommand\csname#2b\endcsname[1][empty]{\relax}\else 
                       \expandafter\renewcommand\csname#2b\endcsname[1][\empty]{\ifx\empty##1\empty
                       	\label{#2-bookmark} 
                              \marginpar [\raggedleft\csname#2\endcsname{{\footnotesize\fbox{#2 working here}}~$\Longrightarrow$}]
                                                {\csname#2\endcsname{$\Longleftarrow$~{\footnotesize\fbox{#2 working here}}}}
                       \else 
                       	\marginpar [\raggedleft\csname#2\endcsname{\ifx\empty##1\empty\else\fbox{\tiny\parbox{8em}{\raggedright##1}}~\fi$\Longrightarrow$}]
                                                {\csname#2\endcsname{$\Longleftarrow$\ifx\empty##1\empty\else~{\tiny\fbox{\parbox{8em}{\raggedright##1}}}\fi}}\fi}\fi

     \expandafter\providecommand\csname#2TD\endcsname {} 
     \ifSuppress\expandafter\renewcommand\csname#2TD\endcsname{\relax}\else
                       \expandafter\renewcommand\csname#2TD\endcsname{\csname#2\endcsname{\fbox{#2 to do}}}\fi

     \expandafter\providecommand\csname#2Bar\endcsname {} 
     \ifSuppress\expandafter\renewcommand\csname#2Bar\endcsname{\relax}\else
                       \expandafter\renewcommand\csname#2Bar\endcsname{\csname#2\endcsname{\scriptsize\XSolidBrush}}\fi

     \expandafter\providecommand\csname#2f\endcsname {} 
     \ifSuppress\expandafter\renewcommand\csname#2f\endcsname[2][]{\relax}\else
      \expandafter\renewcommand\csname#2f\endcsname[2][\empty]{ 
        {\mbox{\csname#2x\endcsname\tiny$\boxtimes$}\marginpar{\hsize1cm\csname#2x\endcsname\fbox{\FnSym\footnotemark}}\relax 
        \footnotetext{\csname#2x\endcsname##2}}}\fi

     \expandafter\providecommand\csname#2e\endcsname {}
     \ifSuppress\expandafter\renewcommand\csname#2e\endcsname[1]{\relax}\else%
      \expandafter\renewcommand\csname#2e\endcsname[1]{%
       \global\EndNotestrue
       \mbox{\scriptsize\csname#2x\endcsname$\boxtimes$}\relax%
       \marginpar{\hsize1cm\csname#2x\endcsname\fbox{\EnSym\endnotemark%
                          \hypertarget{ENmark\thepage-\theendnote}{}~\hyperlink{ENtext\thepage-\theendnote}{{\Colour{blue}$\downarrow$}}}%
       }%
       {
        \def\zz{\noexpand#3}%
        \edef\z{~{[Endnote \theendnote\ %
        on p.\noexpand\hypertarget{ENtext\thepage-\theendnote}{}\thepage%
                    ~\noexpand\hyperlink{ENmark\thepage-\theendnote}{{\noexpand\Colour{blue}$\uparrow$}}]}%
        }%
        \expandafter\endnotetext\expandafter{\z\vspace{2ex}\\ ##1\newpage}%
       }
      }\fi

     \expandafter\providecommand\csname#2n\endcsname {}
     \ifSuppress\expandafter\renewcommand\csname#2n\endcsname[1]{\relax}\else%
      \expandafter\renewcommand\csname#2n\endcsname[1]{%
       \global\EndNotestrue
    \marginpar{{\tiny\endnotemark}\hypertarget{ENmark\thepage-\theendnote}{}~\hyperlink{ENtext\thepage-\theendnote}{}}
       {
        \def\zz{\noexpand#3}%
        \edef\z{~{\zz[Endnote (deferred) 
        from p.\noexpand\hypertarget{ENtext\thepage-\theendnote}{}\thepage%
        ]}%
        }%
        \expandafter\endnotetext\expandafter{\z\vspace{2ex}\\ ##1\newpage}%
       }
      }\fi

     \expandafter\providecommand\csname#2fe\endcsname {} 
     \ifSuppress\expandafter\renewcommand\csname#2fe\endcsname[2][]{\relax}\else 
      \expandafter\renewcommand\csname#2fe\endcsname[2][]{ 
       \def\File{##1}\relax
       \ifx\File\empty\csname#2f\endcsname{##2}\else 
        \global\EndNotestrue 
        \mbox{\scriptsize\csname#2x\endcsname$\boxtimes$}
        \marginpar{\csname#2x\endcsname\fbox{\FnSym\footnotemark}}\relax
        \footnotetext{~\csname#2x\endcsname##2\
                             --- See [\EnSym\endnotemark\hypertarget{ENmark\thepage-\theendnote}{}
                             \kern-.2em\hyperlink{ENtext\thepage-\theendnote}{{\Colour{blue}$\downarrow$}}].}\relax
       { 
         \def\zz{\noexpand#3}
         \edef\z{~{\zz[Endnote~\thefootnote~on~p.\noexpand\hypertarget{ENtext\thepage-\theendnote}{}\thepage
                     ~\noexpand\hyperlink{ENmark\thepage-\theendnote}
                     {{\noexpand\Colour{blue}\kern-0.1em$\uparrow$}]}}
                     {\noexpand\footnotesize\noexpand\newline\noexpand\hspace*{2em} (~from file {\noexpand\tt\File.tex}~)}
         }    
         \expandafter\endnotetext\expandafter{\z~\par\input{##1}\newpage}
        } 
       \fi 
      } 
     \fi 

     \ifSuppress\relax\else\ifBleck\relax\else
      \MarkupsHowtoAdd{\par\csname#2t\endcsname{
       $\backslash$\texttt{#2}$\cdots$\ markups are in \textbf{this} colour\ifx#1..\else\ifx1#1.\else, e.g.\ for #1.\fi\fi
       \ifMarkupsHowtoPrinted\relax\else 
        \global\MarkupsHowtoPrintedtrue 
        \begin{quote}\begin{tabular}{l@{\hspace{2em}}p{.7\linewidth}}
         \multicolumn{2}{l}{\texttt{$\backslash$MakeMarkups\ifx#1.\relax\else[#1]\fi\{#2\}\{{\emph{$\langle$colour command\/$\rangle$}}\}}
         				 --- Defines the macros below:}\\
             & see comments at \texttt{$\backslash$MakeMarkups} definition. \\[1ex]
         \texttt{$\backslash$#2\{$\langle$text$\rangle$\}} & Sets \texttt{$\langle$text$\rangle$} in \texttt{#2}'s colour. \\
         \texttt{$\backslash$#2x} & Changes to \texttt{#2}'s colour (until end of context). \\
         \texttt{$\backslash$#2d\{$\langle$text$\rangle$\}} & Sets \texttt{$\langle$text$\rangle$} in \texttt{#2}'s colour with a strikethrough (i.e.\ delete). \\
         \texttt{$\backslash$#2r\{$\langle$this$\rangle$\}\{$\langle$that$\rangle$\}} &
          Strikes through \texttt{$\langle$this$\rangle$} and inserts \texttt{$\langle$that$\rangle$} (i.e.\ replace). \\
         \texttt{$\backslash$#2f\{$\langle$text$\rangle$\}} & Meta-comment: puts \texttt{$\langle$text$\rangle$} in a \texttt{#2}-footnote with a {\tiny$\boxtimes$} in the main text. \\
         \texttt{$\backslash$#2t\{$\langle$text$\rangle$\}} & Use for meta when  \texttt{$\backslash$#2f} isn't allowed (``Not in outer-par mode.'') \\
         \texttt{$\backslash$#2b[$\langle$optional$\rangle$]} & Marginal pointer, with label for hyper-linking directly there. \\
         \texttt{$\backslash$#2e\{$\langle$text$\rangle$\}} & Puts \texttt{$\langle$text$\rangle$} in a \texttt{#2}-endnote with a (big) $\boxtimes$ in the main text. \\[.5ex]
         \texttt{$\backslash$#2n\{$\langle$text$\rangle$\}} & Like \texttt{$\backslash$#2e}
         except there's no reference from the main text. Good for ``decluttering''
         when you still want to have the footnote- or endnote texts as reminders. \\[.5ex]
         \texttt{$\backslash$#2fe[$\langle$this$\rangle$]\{$\langle$that$\rangle$\}} & Makes a \texttt{$\backslash$#2f\{$\langle$that$\rangle$\}} that refers to a \\
           & \texttt{$\backslash$#2e\{$\langle$contents of file this.tex$\rangle$\}}. \\ 
           & Without the optional argument, acts as \texttt{$\backslash$#2f\{$\langle$that$\rangle$\}}. \\[.5ex]
         \texttt{$\backslash$#2Bar} & Inserts ``burn after reading'' symbol \csname#2Bar\endcsname, meaning
          \begin{quote}\begin{itemize}\setlength\itemsep{0pt}
           \item If yours is the only \csname#2Bar\endcsname\ in this (presumably someone else's) footnote, and you are happy that the footnote has been addressed,
           go ahead and comment-out the whole footnote. (The \csname#2Bar\endcsname\ is their request for you to ``approve and remove''.)
           \item If you are not happy, delete only your \csname#2Bar\endcsname\ and follow-on in the footnote
            (in your colour, i.e.\ with \texttt{$\backslash$#2x}) saying why you are not happy.
           \item If you are happy, but there are others' burn-after-reading symbols as well as yours, just delete yours; the other people have not yet responded.
          \end{itemize}
          \end{quote}
          The idea is that when everyone's happy, the last person will comment-out the meta-text. \\[0.5ex]
         \texttt{$\backslash$#2TD} & Inserts {\csname#2TD\endcsname}\ . \\
        \end{tabular}\end{quote}
       \fi
      }}
     \fi\fi
}

\newif\ifNoGreenRoom
\newcommand\MakeGreenRoom {\ifBleck\relax\else\ifNoGreenRoom\relax\else
\newcommand\NewGRLabel[1] {\OldGRLabel{GreenRoom-##1}} 
 \newcommand\NewGRRef[1] 
 {\expandafter\ifx\csname r@GreenRoom-##1\endcsname\relax\OldGRRef{##1}\else\OldGRRef{GreenRoom-##1}\fi}
 \let\OldGRLabel\label \let\label\NewGRLabel
 \let\OldGRRef\ref \let\ref\NewGRRef
 \hrule
 ~\\\begin{center}\Huge \hypertarget{GreenRoom}{Green Room}
 \end{center}~\\
 \hrule
\fi\fi}

\newcommand\EndGreenRoom  {\ifBleck\relax\else\ifNoGreenRoom\relax\else
\let\label\OldGRLabel
\let\ref\OldGRRef
\fi\fi}

\newif\ifNoEndNotes
\newcommand\MakeEndNotes {\ifBleck\relax\else\ifNoEndNotes\relax\else
 \hrule\vspace{20pt}\begin{center}\ifEndNotes\Huge Endnotes\else\textit{No endnote macros were called in this run.}\fi\end{center}\vspace{20pt}\hrule
 {\parindent 0pt \parskip 2ex \def\enotesize{\normalsize} \def\notesname{} 
 \ifEndNotes\theendnotes\fi} 
\fi\fi}

\newif\ifNoSargasso
\newcommand\MakeSargasso {
 \hypertarget{Sargasso}{}
 \newcommand\NewLabel[1] {\OldLabel{Sargasso-##1}} 
 \newcommand\NewRef[1] 
 {\expandafter\ifx\csname r@Sargasso-##1\endcsname\relax\OldRef{##1}\else\OldRef{Sargasso-##1}\fi}
 \let\OldLabel\label \let\label\NewLabel
 \let\OldRef\ref \let\ref\NewRef
\ifBleck\end{document}\else\ifNoSargasso
\relax
\else
  \hrule
  ~\\\begin{center}\Huge Sargasso
  \end{center}~\\
  \hrule
 \fi\fi
}

\newcommand\EndSargasso  {\ifBleck\relax\else\ifNoSargasso\relax\else
\let\label\OldLabel
\let\ref\OldRef
\fi\fi}

\newcommand\EndDocument {\ifBleck\end{document}\fi} 

\newcommand\Cite[2][\empty] {{\Colour{red}\ifx#1\empty[#2]\else[#2,~#1]\fi}}



%% file: 1-introduction.tex
\section{Introduction}\label{sec:introduction}

Tor uses \emph{circuits} for anonymous communication. Tor circuits are
multi-hop paths through the Tor network carrying application data.  There are
various types of circuits, some of them for navigating the regular Internet,
others for fetching Tor directory information, connecting to onion services or
simply for measurements and testing.  Although all traffic is encrypted, it is
still possible for certain type of adversaries to distinguish Tor circuit types
from each other using a wide array of metadata and distinguishers; such attacks
are known as \emph{circuit fingerprinting}.

Circuit fingerprinting attacks are dangerous because by determining the Tor
circuit type, an adversary immediately gets to shrink the ``world size'' or
``base rate'' of traffic that it needs to consider for \emph{website
  fingerprinting} or \emph{end-to-end correlation} attacks to succeed
\cite{website-fpr-juarez}. A website fingerprinting adversary who learns that
the client is visiting an onion service immediately enjoys a significantly
smaller universe size of web pages that it needs to identify, effectively
reducing the problem to a closed-world setting \cite{analysis-fpring}.

\paragraph{Contributions}

In this paper, we study the circuit fingerprinting problem, isolating and
formalizing it on its own, separately from other fingerprinting problems.  We
start by presenting our threat model (\autoref{sec:model}) and experimental
methodology (\autoref{sec:methodology}). We then revisit and verify previous
results by Kwon et al. \cite{KwonALDD15} (\autoref{sec:vanilla-fpring}) to show
that distinguishing onion service circuits is possible with high accuracy, even
when the application-layer traffic is identical, by exploiting fingerprints of
the onion service \emph{protocol itself}.

We design and implement a versatile adaptive padding framework (\autoref{sec:wtf}) which has since
been merged in the mainline Tor client and deployed to the live network.
Then, in \autoref{sec:design}, we
proceed by analyzing onion handshakes and demonstrating how we can imitate
accurately both their shape and timing characteristics using our adaptive padding framework.
Employing our dummy onion
handshakes as a primitive, we then design and evaluate two distinct padding defense
strategies.

First, we present an intuitive delay-based defense (\autoref{sec:fractional-delay})
which delays clearnet traffic to inject a dummy handshake.
We show that the latency overhead can be reduced at no privacy cost by
delaying only a \emph{fraction} of the traffic and exploiting the base rate
to our advantage.
During the evaluation of this  defense we also find previously unknown subtle
behaviors that can act as circuits fingerprints.

We then present a more advanced zero-delay padding defense that employs
\emph{preemptive circuits}, an optimization mechanism that already exists in
Tor. We provide both an analytic evaluation, giving precise expressions
relating the parameters of the defense to the classification accuracy, as well
as an experimental evaluation, confirming the analytic results, and showing
that an effective defense is possible with no delays and limited bandwidth
overhead.

\paragraph{Tor network}

The Tor network is among the most popular tools for digital privacy and
anonymity. As of September 2020, the Tor network consists of almost 6,500
volunteer-run relays.

In its basic mode, Tor allows its clients to achieve anonymity when connecting
to TCP/IP services. This is achieved by forming a virtual circuit through the
relays of Tor network in an interactive and incremental fashion. This is
typically done by extending the circuit to three hops: the \emph{entry guard};
the \emph{middle}; and the \emph{exit}.

Once a circuit is established, the client, Alice, creates \emph{streams}
through the circuit by instructing the exit to connect to the desired external
Internet destinations. Each pair of relays communicate over a single onion
routing connection that is built using TCP. Application layer protocols rely on
this underlying TCP connection to guarantee reliability and in-order delivery
of application data, using fixed-size data packets (512 bytes) called
\emph{cells}, between each relay. Streams are multiplexed over circuits, which
themselves are multiplexed over connections.

In the case of a website, Alice asks the \emph{exit relay} to connect to the
destination website and the communication between Alice and the website happens
through the resulting circuit. We call the circuit that connects to the
normal web an \emph{exit circuit} and the whole connection a \emph{clearnet
  connection}.

\paragraph{Tor Onion services}

In addition to client anonymity, Tor allows operators to set up anonymous
servers, typically called \emph{onion services}\cite{rend-spec}. This is
achieved by routing communication between the client and the onion service
through a rendezvous point which connects anonymous circuits from the client
and the server.

When a client, Alice, connects to an onion service Bob, she first creates a
directory circuit (in this work we call it \emph{HSDir circuit}) and fetches
the \emph{descriptor} document of the onion service. The descriptor contains
the \emph{introduction points} of the onion service which are also relays on
the Tor network.

Alice opens a \emph{rendezvous circuit} to a randomly sampled relay on the Tor
network, and asks it to become the \emph{rendezvous point} for this
session. Alice then connects to one of the introduction points of Bob using an
\emph{introduction circuit} and instructs Bob to meet at the determined
rendezvous point.
When Bob connects to the given \emph{rendezvous point}, the latter
bridges Alice's and Bob's circuits and the rest of the communication happens
over the spliced rendezvous circuit.

During this work, we call the set of an \emph{HSDir circuit},
\emph{introduction circuit} and \emph{rendezvous circuit} the \emph{onion
  service circuit triplet}. All the communication over Tor circuits described
above occurs using Tor \emph{cells} and each circuit type has its own unique
cell pattern. Recognizing those cell patterns and distinguishing the circuit
purposes from each other is the art of \emph{circuit fingerprinting} that we
explore in this paper.

\paragraph{Related work}

While circuit fingerprinting has been assumed to be possible in the past, Kwon
et al. \cite{KwonALDD15} were the first to formally present the problem, their
methodology and demonstrate successful classifiers against onion service
circuits. They extracted specific features from the circuit cell flow and
trained classifiers in a way that offered very high accuracy. Kwon et
al. focused on entry guard and network-level adversaries launching circuit
fingerprinting attacks, while subsequent research by Jansen et
al.\cite{inside-job} explored circuit fingerprinting by middle nodes for the
purpose of measuring the popularity of specific onion services.

The \emph{circuit fingerprinting} methodology is inspired by the related field
of \emph{website fingerprinting}. Website fingerprinting was first studied in
the context of SSL encrypted traffic\cite{first-wf}, and in recent years there
is increasing focus in fingerprinting websites over Tor\cite{wf-tor,
  PanchenkoMHLWE17, panchenko-website}. Website fingerprinting experiments have
been using machine learning for classification\cite{deepfingerprinting} and
feature selection has been deeply studied \cite{wf-feature-selection, tik-tok}.

In terms of defenses, padding is a well-established technique in the field of
low-latency anonymity networks\cite{adaptive-shmatikov, peekaboo}. The WTF-PAD
adaptive padding system proposed by Juarez et al. \cite{wtfpad} was designed
against website fingerprinting and has been shown to be effective and also
versatile enough to be used for multiple purposes.

Finally, several attacks have been proposed to deanonymize onion services and
their clients using Sybil and other protocol
attacks\cite{trawling, locating-onions, sniper}.


%% file: 2-model.tex
\section{Threat model}\label{sec:model}

Tor aims to minimize the amount of metadata leaked to various external
adversaries. This work concerns \emph{circuit fingerprinting} which allows
attackers to distinguish onion service connections from other types of
connections. In this section, we isolate this problem and
investigate the types of adversaries that can
launch such attacks.
We assume that the attacker is either a
network adversary or a relay adversary. This is because the attacker needs to
be able to extract fine-grained features out of the client's traffic and in
particular distinguish cells from each other.

A \emph{relay adversary} can be part of the Tor network by setting up relays as
part of a \emph{Sybil attack}. This allows the attacker to spectate or even
influence the user's traffic. During this work, the main adversary class we are
concerned about is adversarial guard nodes since they are in position to know
the location of the user and hence have more power if they manage to carry out
a circuit fingerprinting attack. A guard adversary can perform traffic analysis
attacks on the user's traffic since they have full visibility on the Tor
circuit and the cells transferred within. While they cannot see the types or
contents of Tor cells, they can still see the exact number of incoming and
outgoing cells.

A \emph{network adversary} can be the network administrator of the user, their
ISP, or any other intruder between the user and their guard node. A
network-level adversary can collect TCP traces from the user and then use
heuristic algorithms to turn those traces into fine-grained cell sequences. For
the purposes of this research we assume this is possible with high
accuracy\cite{lan-adversary-ian}. The website fingerprinting literature has
been working with the assumption that this is possible with a high-probability
of success\cite{KwonALDD15,ian-wf-lan}; in \autoref{sec:discussion} we discuss
how that assumption ties into the circuit fingerprinting problem.

\paragraph{Separating the problems}

In this work, we tackle the \emph{circuit fingerprinting} problem, which we
believe is a distinct problem, related to, but separate from \emph{website}
fingerprinting.

\begin{itemize}
  \setlength\itemsep{0.06em}
\item In \emph{circuit fingerprinting}, a relay or network adversary
  distinguishes circuit types and purposes in a Tor session. It concerns the
  Tor protocol traffic.

\item In \emph{website fingerprinting}, a relay or network adversary
  distinguishes websites from each other within a Tor session. It concerns the
  application-layer traffic.
\end{itemize}
Separating these problems from each other allows us to isolate them, gain
a greater understanding on how they interact with each other and design better defenses.

\emph{Website fingerprinting} attacks are difficult to defend against because
the universe size of web pages is vast and there is no centralized way to make
them all look alike. On the other hand, the Tor protocol has a smaller
fingerprinting surface and \emph{circuit fingerprinting} defenses can be
designed and deployed with agility through Tor's protocol upgrades.

Leaving any of those two problems unsolved, allows the adversary to leverage it
to solve the other problem. For instance, an adversary who can solve the
\emph{website fingerprinting} problem can use its distinguisher to make
\emph{circuit fingerprinting} easier by classifying all websites found as a
specific type of circuit. Similarly, an adversary who can solve the
\emph{circuit fingerprinting} problem can use its distinguisher to make
\emph{website fingerprinting} easier, by classifying non-website circuits as
irrelevant for the purposes of website fingerprinting.

With the goal of not conflating these problems, we
devise a methodology for studying circuit-fingerprinting in isolation.


%% file: 3-methodology.tex
\section{Experiment methodology}\label{sec:methodology}

This section describes the methodology employed in all experiments performed in
this paper in order to evaluate the various defenses against a circuit
fingerprinting adversary.

\paragraph{Data collection, processing and classification}

Our data consists of 100 real onion addresses, selected based on their
popularity and availability. We locally fetched every home page and lightly
processed them by stripping off any external resources. Finally, we hosted each
cached website on our own server twice: once on a clearnet server, and an
\emph{identical} version on an onion service.
Serving the same content over both types of connections follows the core
foundation of our methodology of keeping the application-layer traffic
identical. The \emph{circuit fingerprinting} problem is solely concerned with
the Tor protocol layer and hence the application-layer traffic should not play
a role in the classification process.

To collect the actual traffic traces we used Tor Browser and automated it using
Selenium \cite{Selenium}.  We patched the Tor client to log fine-grained
details about its operation, in particular log every
incoming and outgoing cell.
During data collection we instructed Selenium to fetch pages off our dataset
multiple times. The amount of times and the set of pages to fetch depends on
the experiment. Between each page fetch we forced Tor to avoid reusing guards
and to create a new set of circuits instead of using already used circuits.

During the website caching procedure, we make sure that every page element is
downloaded on our machine and no external resources need to be loaded when we
browse our cached web pages. This way we eliminate any obstacles
that may arise during a web request, such as broken links or resources that
might return corrupted content (CDN resources, URL redirects, etc.).

After the data collection step, we need to prepare the data for
classification. For this purpose we created a processing script
that reads the Tor log file and keeps track of
all the circuits and their cells.
The script outputs
a dataset for the classifier, depending on the current \emph{classification
scenario}. The latter provides the logic for our processing script to
classify circuits as onion-related or not. Each experiment uses a different
classification scenario, described in the respective sections of
this work.
The generated dataset contains an entry
for each circuit, including its \emph{label}
(whether it is an onion circuit or not),
and a vector of features.

Finally, we pass the processed data to our classifiers.  Depending on the
experiment, the classifier splits the set of circuits into a training set and a
testing set, and then trains itself using the training set. After evaluating
different types of machine learning and deep learning classifiers, we used SVM
and decision tree classifiers because they are robust and their resulting
fingerprints are easy to analyze and comprehend.

The \emph{label} field of the intermediate file is solely used as ground truth
for training purposes, and to measure the final accuracy of the classifier.

\paragraph{Timing analysis methodology}\label{sec:timing-methodology}

While the primary focus of this work are defenses that can imitate the \emph{shape}
of onion circuits, that is their precise
cell sequence, we also perform preliminary analysis on the
timing fingerprints of onion handshakes in \autoref{sec:handshake-timing}.
For this purpose, we enhanced our data collection logic to keep
the times of incoming and outgoing cells and wrote code that analyzes them to
derive statistical data.

We found that working with timing features on the Tor network was complicated
because the network's performance exhibits wild short-term fluctuations due to
its unpredictable user model, overuse, and the lack of congestion
control\cite{congestion-prop}. To reduce the timing noise observed in our
measurements we carried out measurements over multiple days and pruned extreme
performance outliers.

\paragraph{Classification Scenarios}

Similar to the website fingerprinting world, our experiments also carry the
concept of \emph{open} and \emph{closed world}.
We call an experiment
\emph{open world} when the hosts in the training URL dataset are different from
the hosts in the test URL dataset.
On the other hand, we call an experiment to be \emph{closed world} when the
training URL dataset and the test URL dataset contain sessions to the same
destinations. It is important to note here, that in \emph{closed world}
scenarios we never share the exact same session between training and test sets;
instead we share different connection sessions to the same destination.

In the context of this distinction, we should point out that isolating the two
problems led us to an interesting observation:
while \emph{open world} scenarios have been
traditionally much harder to a \emph{website} fingerprinting attacker, this
asymmetry was not present in our results.
The reason is that, in the isolated \emph{circuit} fingerprinting problem,
the fingerprints lie in the Tor protocol itself and not the
application-layer content. In other words, an adversary learns to
identify \emph{any} onion circuit, not those to a particular website.

We also employ two types of scenarios based on the number of websites they
involve: a \emph{single-site} scenario contains a single destination hosted
both in clearnet and as an onion. The classifier is called to classify whether
a visit to the destination happened over clearnet or over onion. This is easier
for the classifier and it is meant to model a future hypothetical world where
\emph{website fingerprinting} has been solved and all application-layer traffic
is identical.

A \emph{multi-site} scenario contains multiple destinations hosted both in
clearnet and as an onion. The classifier is called to decide whether a visit to
the destination happened over clearnet or over onion without knowing the actual
destination. This is meant to model the current world where all websites look
different and can produce uncertainty to the adversary.

We pick the classification scenarios that suit each experiment. Specifically, a
\emph{single-site closed} scenario is good for demonstrating defenses since
classifying a single known website is the easiest scenario for the
adversary. On the other hand, a \emph{multi-site open} scenario is good for
demonstrating attacks since classifying multiple unknown websites is the
hardest scenario for the adversary.


%% file: 4-vanilla.tex
\section{Evaluating vanilla circuit fingerprinting}\label{sec:vanilla-fpring}
We start off our work by examining whether the claims of Kwon et al.
\cite{KwonALDD15} are still valid if we restrict classification to the circuit
fingerprinting problem \emph{alone}. In particular, we want to examine whether
rendezvous and introduction circuits can be differentiated from all other
circuits, when the application-layer traffic is identical for both clearnet and
onion connections. This is important for ensuring that the information leak
lies in the onion service protocol \emph{itself}.

To perform this experiment, we employ two classifiers: first, a classifier
which separates \emph{introduction} circuits from all other types, and second,
a classifier which separates \emph{rendezvous} circuits from all other types.
These classifiers are similar to those of previous research \cite{KwonALDD15},
where all circuit types apart from the target type are considered to be
noise. Dataset sizes are available in \autoref{sec:dataset-sizes}.

We used the \emph{complete} cell sequence (an array of
$-1$ and $+1$ elements) as a feature vector for classification.
This is possible since cells are large and onion websites are generally
small, fitting in a few hundred cells.
Note that, since cells have fixed size, any shape-related feature
(eg. number or percentage of incoming/outgoing cells)
can be deduced from this feature vector;
we tried including such features explicitly but it made no difference to
the classifier.
Finally, following \cite{KwonALDD15}, we included the lifetime of the circuit as a feature.

For this experiment we used 100 websites from our dataset.  In contrast to
\cite{KwonALDD15}, we used the same websites served over both clearnet and
onion connections, in order to restrict classification to the circuit
fingerprinting problem alone.  In both the \emph{closed} world and \emph{open}
world setting, our classifiers performed very well having an average
accuracy of 98-99\%.
This shows that circuit fingerprinting is possible
without relying on the traffic of the websites themselves; in other words, the
traffic pattern produced by the \emph{protocol itself} can be used as a
fingerprint.

Even more importantly, since the traffic pattern of the website is not exploited
in the attack, circuit fingerprinting is possible with high accuracy even in the
\emph{open world} model. An adversary with no prior knowledge of a website can
infer whether the website is accessed as a clearnet or onion service. In
that sense, circuit type fingerprinting is an inherently ``closed world''
problem, as there are only a handful of different circuit types.


%% file: 5-wtf-infrastructure.tex
\section{An adaptive padding framework for Tor}\label{sec:wtf}

After confirming that fingerprinting onion connections is indeed
possible, we started experimenting with padding defenses by designing an
extensible adaptive padding framework based on WTF-PAD \cite{wtfpad}. We then
implemented it and our padding framework has been merged into upstream Tor
since the \emph{0.4.0.1-alpha} release. We have also written detailed developer
documentation that can be used by researchers to design and experiment with
padding defenses in Tor\cite{tor-circ-padding-docs}. Since its creation, our
framework has been used to protect against website fingerprinting\cite{tobias}
and we envision that it can also be used to build defenses against guard
discovery and other correlation attacks\cite{guardsets}.

Using the Tor padding framework we can schedule arbitrary circuit-level padding
patterns to overlay on top of non-padding traffic;
such padding can occur between clients and relays at any hop of the
client's circuits. Both parties need to support the same padding mechanisms for
the system to function. We added a \emph{padding negotiation cell} to the Tor
protocol that clients can use to negotiate with relays which padding patterns
should be used.

Padding is performed by \emph{padding machines} which have \emph{finite state}.
Every state specifies a different form of padding style, or stage of
padding, in terms of inter-packet timings and total packet counts.
Padding state machines are implemented by filling in a C structure,
which specifies the transitions between padding states based on various events,
probability distributions of inter-packet delays, and the conditions under
which padding machines should be applied to circuits.

This compact C structure representation is designed to function as a
microlanguage which can be compiled down into a bitstring. This allows the
padding machine's parameters to be tuned using optimization methods
such as gradient descent, or generative adversarial networks. When performing
such an optimization search, each padding machine can have a fitness function,
which allows researchers to compare padding machines for relative
effectiveness\cite{tobias}.

The event driven, self-contained nature of this framework is also designed to
make evaluation both expedient and rigorously reproducible.

\paragraph{Minimizing delays in defenses}

One of the goals of our padding-based defenses, is that we should avoid
delaying traffic if possible. Given the low-latency model of Tor, we aim to
avoid additional latency by delaying useful traffic as part of our
defenses. This means that we want to achieve our security goals only by
inserting padding traffic.

We envision that extra padding traffic can scale well in the future since
bandwidth costs decrease over time and throughput increases and hence the
bandwidth is not our most scarce resource. On the other hand, communication
speed between computer links will improve much more slowly.

\paragraph{Padding machines for circuit fingerprinting}

After deploying our framework to the live Tor network, we tested our framework
by designing padding machines that hide introduction circuits,
based on the fingerprints identified in \cite{KwonALDD15} and \autoref{sec:vanilla-fpring}.
Our deployed machines aimed to make introduction circuits look like
\emph{directory circuits} fetching \emph{directory information}. The machines
obfuscate the handshake sequence of introduction circuits by carefully adding
cells to the right place. The padding was initiated by the client and went up
to the middle relay, at which point the middle relay replied with padding of
its own.

We built and
deployed these machines on the live Tor network;
in doing so
we verified that
adaptive padding works and solved various engineering hurdles and issues that
arose. We also hardened our framework and
improved its correctness which is fundamental for future research on this
topic.

On the other hand, we also build a similar machine aimed at hiding rendezvous circuits,
which however failed its goal. Repeating the experiments of \autoref{sec:vanilla-fpring}
with this machine enabled, we found that the classification accuracy  was only
slightly reduced. Moreover, we also discovered that HSDir circuits are
fingerprintable and can be used to distinguish between
onion and clearnet connections.

The main takeaway from this experiment is that hiding just one circuit of the
entire onion service setup is not sufficient to hide the entire onion
connection. In the following section, we will explore this more deeply.

%% file: 6-padding-strategies.tex
\begin{figure}
  \tikzset{>=latex}
  \begin{tikzpicture}[every node/.style={scale=0.81}]
      \tikzset{
        title/.style= {align=left,minimum width=0.7cm,minimum height=0.5cm},
        block/.style= {draw, rectangle, align=center,minimum width=0.7cm,minimum height=0.5cm},
      }

      \node [block] at (10pt,0pt)  (hsdir) {HSDir fetch};

      \node [block, below right of=hsdir] at (60pt, -1pt) (est_rend) {Establish rendezvous};
      \node [block, right =2cm of est_rend] (rend2) {Rendezvous completed};

      \node [block, below right of=est_rend] at (125pt, -15pt) (intro) {Introduction};

      \draw [dashed] (-20pt, 0pt) -- (hsdir.west);
      \draw [dashed] (hsdir.east) -- (250pt, 0pt);

      \draw [dashed] (-20pt, -17pt) -- (est_rend.west);
      \draw [dashed] (est_rend.east) -- (rend2.west);

      \draw [dashed] (-20pt, -32pt) -- (intro.west);
      \draw [dashed] (intro.east) -- (250pt,-32pt);
  \end{tikzpicture}

  \caption{Onion handshake flow}
  \label{fig:onion-handshake}
\end{figure}
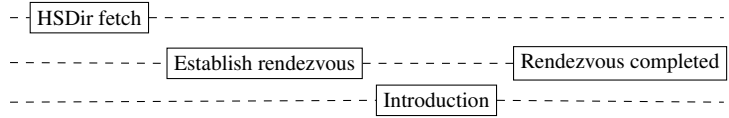

\begin{figure}
  \tikzset{>=latex}
  \begin{tikzpicture}
	\draw (0,0) -- (0,4);
	\draw (2,0) -- (2,4);
	\draw (4,0) -- (4,4);
	\draw (6,0) -- (6,4);
	\draw (8,0) -- (8,4);

	\draw [dotted, arrows={-triangle 45}]
	(0, 4) --  node [midway, fill=white, text centered]
	{ \scriptsize EXTEND } (4, 3.75);

	\draw [dotted, arrows={triangle 45-}]
	(0, 3.25) --  node [midway, fill=white, text centered]
	{ \scriptsize EXTENDED } (4, 3.5);

	\draw [dotted, arrows={-triangle 45}]
	(0, 3.0) -- node [midway, fill=white, text centered]
	{ \scriptsize EXTEND } (6, 2.75);

	\draw [dotted, arrows={-triangle 45}]
	(6, 2.5) -- node [midway, fill=white]
	{ \scriptsize EXTENDED } (0, 2.25);

    \draw[dashed, arrows={-triangle 45}] (0, 2.0) -- node [midway, fill=white, text centered] {\scriptsize EXTEND} (8, 1.75);
    \draw (0, 2.0) -- (3.4, 1.89);

    \draw[dashed, arrows={-triangle 45}] (8, 1.5) -- node [midway, fill=white, text centered] {\scriptsize EXTENDED} (0, 1.25);
    \draw (3.3, 1.35) -- (0, 1.25);

    \draw[dashed, arrows={-triangle 45}] (0, 1.0) -- node [midway, fill=white, text centered] {\scriptsize INTRODUCE1} (8, 0.75);
    \draw (0, 1) -- (3.20, 0.90);

    \draw[dashed, arrows={-triangle 45}] (8, 0.5) -- node [midway, fill=white, text centered] {\scriptsize INTRO\_ACK} (0, 0.25);
    \draw (2.85, 0.34) -- (0, 0.25);

	\node at (0,4.25) {\scriptsize Client};
	\node at (2,4.25) {\scriptsize Guard};
	\node at (4,4.25) {\scriptsize Middle 1};
	\node at (6,4.25) {\scriptsize Middle 2};
	\node at (8,4.25) {\scriptsize Intr. Point};
  \end{tikzpicture}

  \caption{Cell sequence of an introduction handshake}
  \label{fig:intro-circuit-cell-flow}
\end{figure}
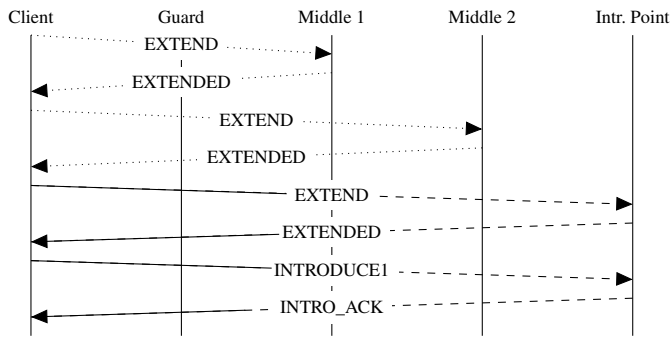

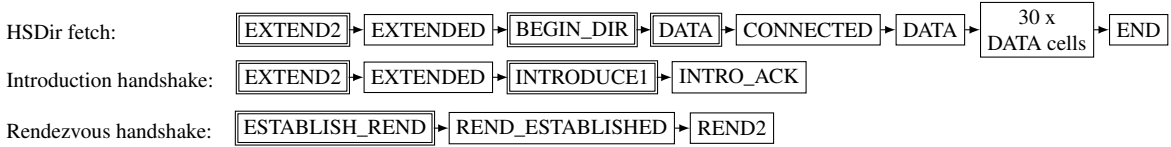
\begin{figure*}
  \tikzset{>=latex}
  \begin{tikzpicture}[every node/.style={scale=0.81}]
      \tikzset{
        title/.style= {align=left,minimum width=0.7cm,minimum height=0.5cm},
        block/.style= {draw, rectangle, align=center,minimum width=0.7cm,minimum height=0.5cm},
        rblock/.style={draw, shape=rectangle,double,double distance=0.02cm,align=center,minimum width=1cm,minimum height=0.5cm},
      }
      \node [title] at (0,-0.03) (title) {HSDir fetch:};
      \node [rblock] at (3.08,0) (extend2_1) {EXTEND2};
      \node [block, right =0.2cm of extend2_1] (extended_1) {EXTENDED};
      \node [rblock, right =0.2cm of extended_1] (begin_dir) {BEGIN\_DIR};
      \node [rblock, right =0.2cm of begin_dir] (data_1) {DATA};
      \node [block, right =0.2cm of data_1] (connected) {CONNECTED};
      \node [block, right =0.2cm of connected] (data_2) {DATA};
      \node [block, right =0.2cm of data_2] (data_all) {30 x \\ DATA cells};
      \node [block, right =0.2cm of data_all] (end) {END};

  Paths
      \path[draw,->] (extend2_1) edge (extended_1)
                  (extended_1) edge (begin_dir)
                  (begin_dir) edge (data_1)
                  (data_1) edge (connected)
                  (connected) edge (data_2)
                  (data_2) edge (data_all)
                  (data_all) edge (end)
                  ;
  \end{tikzpicture}

  \tikzset{>=latex}
  \begin{tikzpicture}[every node/.style={scale=0.81}]
      \tikzset{
        title/.style= {align=left,minimum width=0.7cm,minimum height=0.5cm},
        block/.style= {draw, rectangle, align=center,minimum width=0.7cm,minimum height=0.5cm},
        rblock/.style={draw, shape=rectangle,double,double distance=0.02cm,align=center,minimum width=1cm,minimum height=0.5cm},
      }
      \node [title] (title)  at (0,-0.03) {Introduction handshake:};
      \node [rblock] at (2.43,0) (extend2_1) {EXTEND2};
      \node [block, right =0.2cm of extend2_1] (extended_1) {EXTENDED};
      \node [rblock, right =0.2cm of extended_1] (intro1) {INTRODUCE1};
      \node [block, right =0.2cm of intro1] (intro_ack) {INTRO\_ACK};

  aths
      \path[draw,->] (extend2_1) edge (extended_1)
                  (extended_1) edge (intro1)
                  (intro1) edge (intro_ack)
                  ;
  \end{tikzpicture}
  \vspace{5px}

  \tikzset{>=latex}
  \begin{tikzpicture}[every node/.style={scale=0.81}]
      \tikzset{
        title/.style= {align=left,minimum width=0.7cm,minimum height=0.5cm},
        block/.style= {draw, rectangle, align=center,minimum width=0.7cm,minimum height=0.5cm},
        rblock/.style={draw, shape=rectangle,double,double distance=0.02cm,align=center,minimum width=1cm,minimum height=0.5cm},
      }

      \node [title] at (0,-0.03) (title) {Rendezvous handshake:};
      \node [rblock] at (3,0) (establish_rend) {ESTABLISH\_REND};
      \node [block, right =0.2cm of establish_rend] (rend_established) {REND\_ESTABLISHED};
      \node [block, right =0.2cm of rend_established] (rend2) {REND2};

      Paths
      \path[draw,->] (establish_rend) edge (rend_established)
                  (rend_established) edge (rend2)
                  ;
  \end{tikzpicture}

  \caption{Cell sequences of HSDir / Intro / Rendezvous handshakes after the circuit has finished extending to \emph{Middle 2}}
  \label{fig:dummy-patterns}
\end{figure*}

\section{Onion Handshakes}\label{sec:design}

Our previous experiments have shown that we must convincingly imitate the
onion service setup protocol in its entirety.  Onion connections are
fundamentally different from clearnet connections in that they involve three
extra handshakes: the HSDir fetch, the introduction and the rendezvous. If the
adversary can fingerprint one of those onion handshakes, they immediately
fingerprint the entire connection.

These three handshakes occur on separate circuits and follow the timeline
displayed in \autoref{fig:onion-handshake}.  First the onion descriptor is
fetched from an HSDir, followed by the first half of the rendezvous handshake,
followed by the introduction and finally the second half of the
rendezvous.  After all handshakes are complete, the rendezvous
circuit is used for transferring application data, while the other two circuits
are destroyed.

To avoid distinguishing onion from clearnet connections,
we need to generate \emph{dummy} HSDir, introduction and rendezvous handshakes.
We refer to these three handshakes combined as a \emph{dummy onion handshake}
or \emph{dummy triplet}.
In this section we describe how to generate dummy handshakes using
adaptive padding, while maintaining both the \emph{cell sequence} and the
\emph{timing} characteristics of the corresponding real ones.  In the
following sections, we use onion handshakes as a building block to design two
different defense strategies.

\paragraph{Cell sequence}

All onion handshakes have deterministic cell sequences
(\autoref{fig:dummy-patterns}) which allows us to accurately recreate them using
padding machines. For instance, consider a real
introduction handshake displayed in
\autoref{fig:intro-circuit-cell-flow}. First, the circuit is
extended to \emph{Middle 2} before the user connection arrives
(such \emph{preemptive} circuits improve performance;
they are displayed as \emph{dotted lines} and further discussed in
\autoref{sec:pcp}).
Then, the introduction involves two
round-trips to the Intro Point, one to extend the circuit and another one to
establish the intro point.

To generate a dummy introduction handshake, we use a padding machine running on
\emph{Middle 1}. First, the circuit is preemptively extended to \emph{Middle 2}
as in the paragraph above.  Then, for each of the two round-trips to the Intro
Point, we instruct the client padding machine to send a padding cell to
\emph{Middle 1}, to which the padding machine of \emph{Middle 1} responds with
its own padding cell. Using this method, we effectively fake the communication
between the client and the Intro Point.

The same technique can be used for HSDir and rendezvous handshakes; the exact
sequences 
are shown in \autoref{sec:cell-diagrams}. Padding machines using our
framework that implement a dummy introduction handshake can be seen in
\autoref{sec:padding-machines}.

\paragraph{Timing analysis}\label{sec:handshake-timing}

To successfully generate dummy handshakes, our padding machines need to properly
imitate the delays of real ones, both between cells of the same circuit and
between different circuits.
Revisiting the introduction handshake of \autoref{fig:intro-circuit-cell-flow},
the only difference between a real and a dummy handshake is the communication
bentween Middle~1 and the Intro Point, which is faked by the padding machine.
(dashed lines in \autoref{fig:intro-circuit-cell-flow}).
Hence the machine in Middle~1
needs to emulate the delay between the mement EXTEND is sent and EXTENDED 
is received from the Intro Point (and similarly for
INTRODUCE1/INTRODUCE\_ACK).
Since the time to process the handshake is neglible compared to network times,
we essentially need to imitate the round-trip delay between \emph{Middle 1} and the Intro Point.

Delays of the same nature are exhibited on HSDir and rendezvous circuits. For the
latter, we need to imitate both the round-trip
between Middle 1 and Rend, but also the delay between REND\_ESTABLISHED and
REND\_2 (\autoref{fig:rendezvous-circuit-cell-flow}).

Finally, note that the communication in each circuit starts when the previous
one finishes (\autoref{fig:onion-handshake}). Hence, if \emph{inter-cell} delays are
properly imitated, the \emph{inter-circuit} delays will also match with no extra
effort.


\paragraph{Faking inter-cell delays}\label{sec:timing-experiments}

\begin{figure}
	\centering
	\includegraphics[width=\columnwidth]{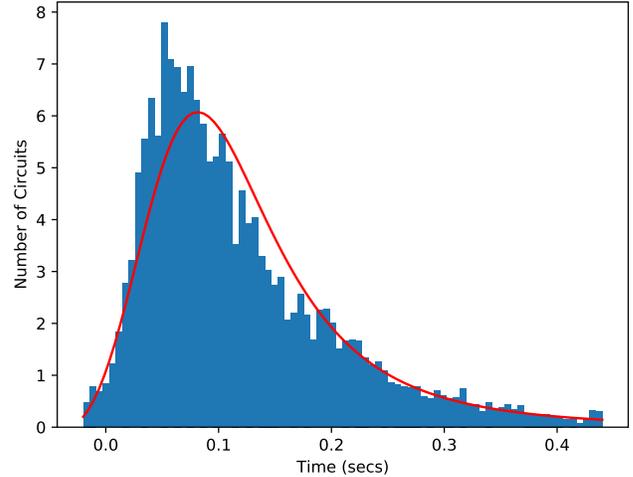}
	\caption{Normalized histogram of round-trip delay between Middle 1 and Intro Point. The red line is a fitted PDF of the LogLogistic(alpha=3.19, beta=0.14) probability distribution}
	\label{fig:loglogistic-fit}
\end{figure}

Using our methodology from \autoref{sec:timing-methodology}, we collected a
dataset of the relevant round-trip delays for each pair of cells in onion
handshakes. We fitted the resulting distributions in a variety of well-known
ones, and found that they can be most accurately approximated by
a Log-Logistic model \cite{log-logistic}.
An example can be seen in
\autoref{fig:loglogistic-fit} for the round-trip delay between INTRODUCE1 and
INTRODUCE\_ACK and a dataset of 10,000 circuits.

To experimentally verify our findings we created a dataset of dummy onion
handshakes whose inter-cell delays were sampled from the corresponding fitted
Log-Logistic distributions. We then ran an experiment
against cell traces of real data, in which the +1/-1 feature vector is replaced
by the exact inter-cell delays. The resulting accuracy was at most 0.52 in all types of
circuits, showing that our ML adversary could not distinguish dummy onion
handshakes from real ones, despite the availability of detailed inter-cell
timing information.

The next question is \emph{when} exactly to perform these dummy
handshakes, and on \emph{which circuits}. We answer this question in two
different ways, giving rise to two padding strategies discussed in the next
sections. First,
we inject dummy handshakes only in the clearnet traffic
while delaying the actual data. Second, we
avoid delays by injecting dummy handshakes \emph{preemptively},
to both clearnet and onion traffic.


%% file: 7-high-latency.tex
\section{The fractional-delay strategy}\label{sec:fractional-delay}

We first present a simple and intuitive strategy which pads clearnet
connections to make them look like onion connections: when Tor receives a
clearnet request from the user, we \emph{pause} that request and inject a dummy
onion handshake, as described in \autoref{sec:design}.  When the dummy handshake
is complete, the clearnet request is \emph{resumed} by sending the application
data over the rendezvous-handshake circuit.

More precisely, when the user tries to access a destination on an exit circuit,
the Tor client delays the request and promptly creates two new two-hop
circuits, extended to random middle nodes%
\footnote{ Note that if preemptive circuits are available (see
  \autoref{sec:pcp}), they can be used for this purpose. }. We then use the
techniques from \autoref{sec:design} to inject a dummy HSDir handshake to the
first circuit and a dummy introduction handshake to the second
circuit. Finally, the actual exit circuit for the clearnet connection is
created, but we first instruct our padding machine to inject a dummy rendezvous
handshake to the circuit \emph{before} the actual traffic. After the dummy
rendezvous handshake is complete, the application data is sent as usual.  This
procedure, overall, creates three circuits whose traffic looks
indistinguishable from a real HSDir, introduction and rendezvous triplet.

This technique makes clearnet connections indistinguishable from onion ones
with respect to both the shape (size and cell sequence) and the timing of their
traffic. The reason is that the dummy handshake itself is indistinguishable
from a real one wrt both its shape and timing (\autoref{sec:design}), and
moreover, it is injected at the exact moment where the real handshake appears
in an onion connection.

In the remaining of this section, we first experimentally evaluate the
effectiveness of this strategy. Then, we introduce \emph{fractional-delay}, a
simple improvement of this strategy, which reduces latency without sacrificing
privacy by padding only a fraction of clearnet connections.

\subsection{Experimental evaluation}\label{sec:evaluation-delay}

In this section we perform a thorough evaluation of this strategy, which
reveals previously unknown aspects of the network that can act as fingerprints.

As previously discussed, the need to delay traffic makes this defense
unimplementable in the currently deployed framework of \autoref{sec:wtf}. As a
consequence, we evaluate the defense by \emph{simulating} its padding as part
of our experiment pipeline, as follows.

We call \emph{fake} a circuit with only dummy handshakes sent on it, and
\emph{padded} a circuit with both dummy and real data.
First, for each \emph{exit circuit} in our dataset, we create two new circuits,
a \emph{fake HSDir} and a \emph{fake introduction} one, simulating those used
in the defense. We inject a dummy HSDir handshake to the first circuit and a
dummy introduction handshake to the second one, and add them to the
dataset. Then, we pad the exit circuit itself, injecting a dummy rendezvous
handshake right before the application data.  Note that this defense only adds
padding to clearnet connections, the onion traffic is not modified at all.

\begin{table*}[t]
	\caption{Accuracy results (delay-based defense)}
	\centering
	\begin{tabular}{|c|c|c|c|c|c|c|}
	\hline
	\multirow{2}{*}{} & \multicolumn{2}{c|}{\textbf{Fake-HSDir vs HSDir}} & \multicolumn{2}{c|}{\textbf{Fake-Intro vs Intro}} & \multicolumn{2}{c|}{\textbf{Padded-Exit vs Rend}} \\
	\cline{2-7}
	& \begin{tabular}{@{}c@{}}Dec. Tree\end{tabular} & \begin{tabular}{@{}c@{}}SVM\end{tabular} & \begin{tabular}{@{}c@{}}Dec. Tree\end{tabular} & \begin{tabular}{@{}c@{}}SVM\end{tabular} & \begin{tabular}{@{}c@{}}Dec. Tree\end{tabular} & \begin{tabular}{@{}c@{}}SVM\end{tabular} \\
	\hline
	Multi-Closed & 0.51 & 0.51 & 0.49 & 0.49 & 0.62 & 0.79 \\
	\hline
	Multi-Open & 0.52 & 0.52 & 0.50 & 0.50 & 0.50 & 0.49 \\
	\hline
	\end{tabular}
	\label{Tab:experiment3_accuracy}
\end{table*}

In contrast to previous experiments, using this defense both clearnet and onion
connections have three circuits. As a consequence, in our experiments we
evaluate the adversary's ability to distinguish each of the three pairs of
circuits independently. More precisely we employ three classifiers: first, a
classifier which separates \emph{fake HSDir} from \emph{real HSDir} circuits.
Second, a classifier which separates \emph{fake introduction} from \emph{real
introduction} circuits, and third a classifier which separates \emph{padded
exit} from \emph{rendezvous} circuits.

In terms of the feature space of this experiment, we focus on the entire cell
sequence of the circuits involved (encoded as an array of \emph{-1} and
\emph{+1} elements). Since all Tor cells have the same size, by using the
entire cell sequence we also inherit any other features related to bandwidth.

\paragraph{Evaluation results}

We first evaluated these classifiers in a \emph{multi-open} and a
\emph{multi-closed} dataset consisting of ten different websites (dataset sizes
are shown in \autoref{sec:dataset-sizes}). The results, displayed in
\autoref{Tab:experiment3_accuracy}, show that the defense was effective in
completely hiding HSDir and introduction circuits in all scenarios.  For
rendezvous circuits, the defense was also successful in the \emph{open world}
scenario, where the classifier had not seen the same sites during training.  In
the \emph{closed world} scenario, however, the adversary has non-negligible
accuracy and precision (Appendix, \autoref{Tab:experiment3_statistics}).

\begin{figure}
	\centering
	\includegraphics[width=\columnwidth]{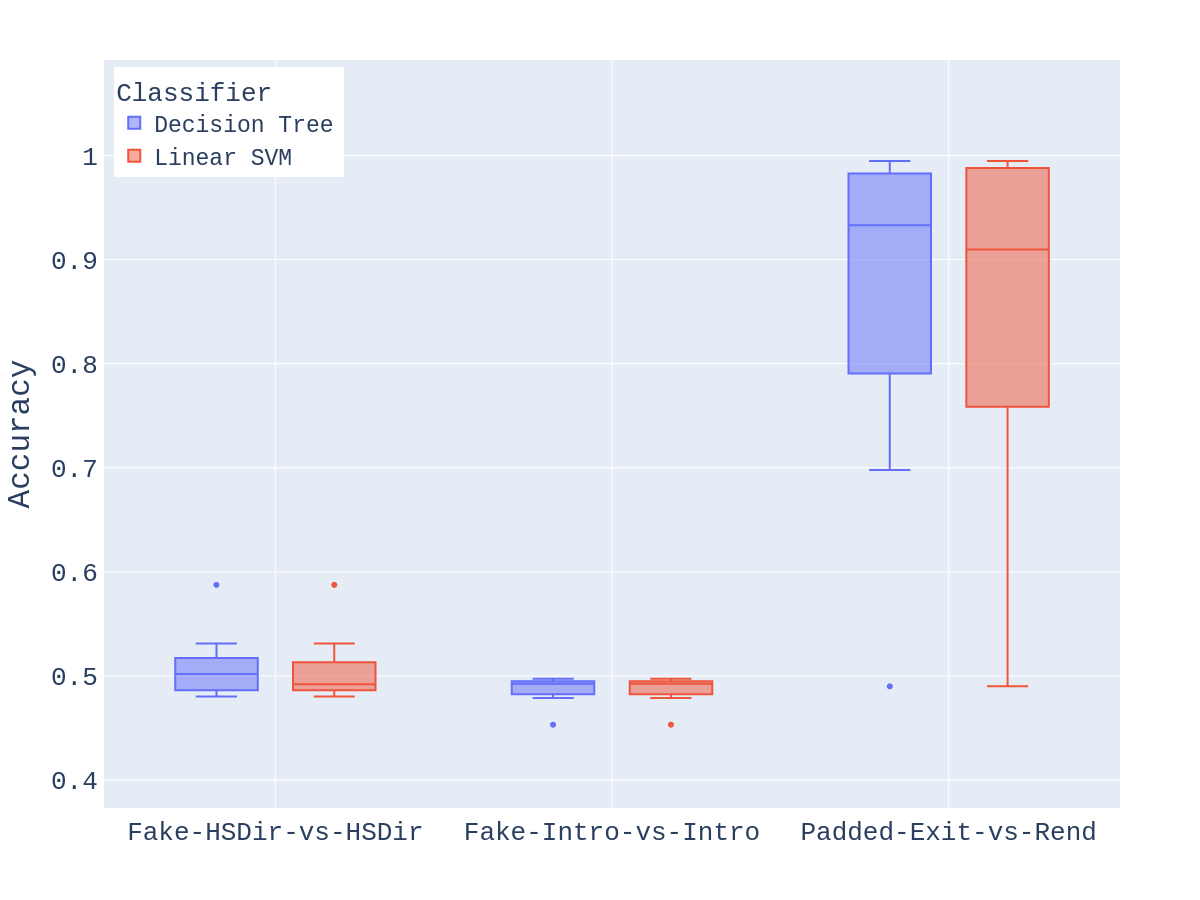}
	\caption{Single site accuracy results (delay-based defense)}
	\label{fig:experiment3_single_boxplot}
\end{figure}

To get a better understanding we then performed the same experiment in the
\emph{single-site} case, for each one of the same ten websites individually. The
results, displayed in \autoref{fig:experiment3_single_boxplot}, highlight the
same behavior: HSDir and introduction circuits are hidden, but rendezvous
circuits can be distinguished from padded exit ones for most websites with high
accuracy (reaching even $100\%$).

This behavior is quite surprising, since in principle the cell sequence of
padded exit circuits is identical to those of rendezvous circuits. By examining
the fingerprints used by our classifiers we managed to identify a surprising
find that causes this behavior. In all experiments we tried to keep the
\emph{application-layer} data identical, by having identical content served
from web servers with identical configurations.  Still, the actual
application-layer traffic did not seem to match in terms of cell sequences. We
found the following two major issues that caused this unexpected phenomenon.

\paragraph{Cell packing difference}

The biggest fingerprint exploited by our classifiers is the fact that
\emph{cell packing is different} between clearnet and onion connections.
Specifically, in the onion case, the cell packing is done by the \emph{onion
  server}, who typically receives the data from a web server running on the
\emph{same machine}.  This means that the web server transmits
application-layer data to the onion server via a high-bandwidth localhost
connection, allowing the provider to optimally package that application-layer
traffic into cells.

However, in the clearnet case, the cell packing is done by the \emph{exit
node}, who receives the data from the destination web server over a remote
connection, possibly from a different continent. As a consequence,
unpredictable networks delays are added to the application-layer traffic
before it reaches the exit node, causing cells to be packed more loosely, but
also a higher \emph{variance} in the resulting number of cells.

The number of cells observed in our experiments for a specific website are
shown in the Appendix (\autoref{fig:cell_dif_exp3}).  Rendezvous circuits are
clearly packed more densely, but also have less variance in their cell count.

\paragraph{Latency influences the application-layer data}

Another, more subtle but present fingerprint was the fact that the inherent
latency difference between exit and rendezvous circuits (due to the different
number of hops) was creating fingerprintable patterns on application-layer
cell sequences. In particular, we found that, for complex websites, the Tor
Browser's scheduling of the various page resources (images, scripts, etc) was
sensitive to latency differences, causing, for instance, images to be
predictably loaded in a different order over clearnet connections than over
onion ones. Note that Tor Browser opens multiple TCP connections to fetch
resources, which is important for this fingerprint to appear.
However, the exact reason why the scheduling order is so predictably different
is unclear and left as future work.

Finally, it should be noted that the two fingerprints described in the previous
sections both lie in the application data, not the onion service protocol
itself. As a consequence, they highly depend on the traffic patterns of each
individual website, which is why the defense was effective in the open-world
scenario, but ineffective in the closed-world (\autoref{Tab:experiment3_accuracy}).

\subsection{Bypassing the hurdles}\label{sec:experiment4}

\begin{table*}[t]
	\caption{Multi site accuracy results (delay-based strategy over localhost)}
	\centering
	\begin{tabular}{|c|c|c|c|c|c|c|}
	\hline
	\multirow{2}{*}{} & \multicolumn{2}{c|}{\textbf{Fake-HSDir vs HSDir}} & \multicolumn{2}{c|}{\textbf{Fake-Intro vs Intro}} & \multicolumn{2}{c|}{\textbf{Padded-Exit vs Rend}} \\
	\cline{2-7}
	& \begin{tabular}{@{}c@{}}Dec. Tree\end{tabular} & \begin{tabular}{@{}c@{}}SVM\end{tabular} & \begin{tabular}{@{}c@{}}Dec. Tree\end{tabular} & \begin{tabular}{@{}c@{}}SVM\end{tabular} & \begin{tabular}{@{}c@{}}Dec. Tree\end{tabular} & \begin{tabular}{@{}c@{}}SVM\end{tabular} \\
	\hline
	Multi-Closed & 0.49 & 0.49 & 0.49 & 0.49 & 0.48 & 0.48 \\
	\hline
	Multi-Open & 0.50 & 0.50 & 0.50 & 0.50 & 0.50 & 0.50 \\
	\hline
	\end{tabular}
	\label{Tab:experiment4_accuracy}
\end{table*}

\begin{figure}
	\includegraphics[width=\columnwidth]{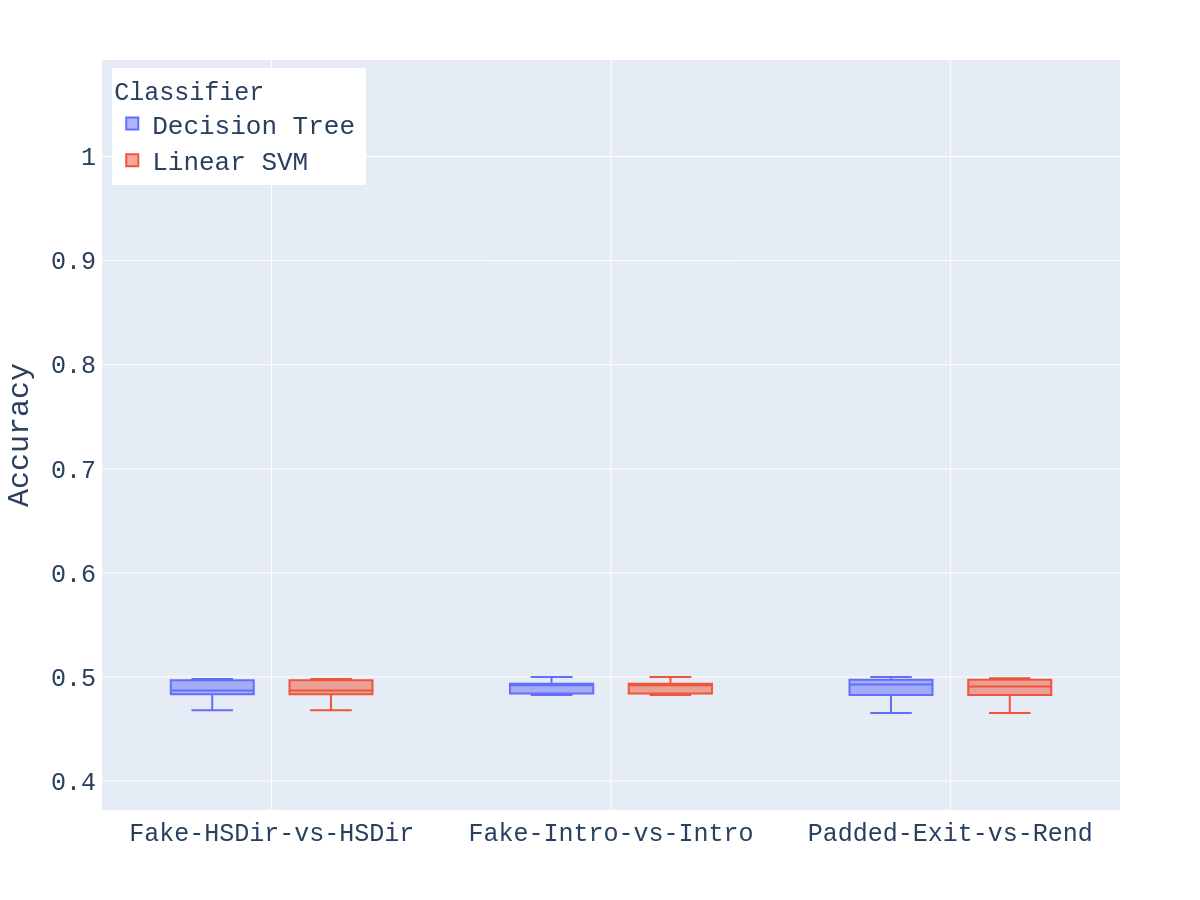}
	\caption{Single site accuracy results (delay-based strategy over localhost)}
	\centering
	\label{fig:experiment4_single_boxplot}
\end{figure}

The two fingerprints described in the previous paragraphs are fundamental,
affecting the majority of onion service setups. Defending against them is
possible but challenging.

With respect to the cell packing fingerprint, we believe that onion servers can
work around it by placing the underlying service endpoint on a remote host. It
is also possible to enhance our padding machines so that they emit additional
padding cells on rendezvous circuits to simulate the non-optimal cell packing
of exit circuits. With respect to the Tor Browser's scheduling fingerprints, a
defense could be to limit the number of concurrent connections to a single one
(however that would negatively impact performance), or to improve the
scheduling algorithm to make it insensitive to latency differences.

Although such defenses are an interesting topic for future work, they are
beyond the scope of this paper, which focuses on the onion service protocol
itself and not on the application layer. For this reason, in the remaining of
this paper we work around the above issues and evaluate our defenses under the
assumption that such issues are not present. In particular, we employ a Tor
network running completely on a single machine, using the Chutney Tor emulator,
avoiding the problem of cell packing differences and ensuring that the
application-layer cell length is the same in both clearnet and onion
connections. Furthermore, we use wget instead of the Tor Browser to avoid the
fingerprints caused by scheduling issues.

With these workarounds applied, we repeated the experiment of
\autoref{sec:evaluation-delay}. The results for the multi-open and multi-closed
scenarios, displayed in \autoref{Tab:experiment4_accuracy}, show that the
defense is now effective in hiding all types of circuits, with an accuracy
close to 0.5 even in the multi-closed case.  Similarly, the results for the 10
individual single-site scenarios, displayed in
\autoref{fig:experiment4_single_boxplot}, also show low classification accuracy
for all circuit types.

\subsection{Fractional delay: exploiting the base rate}
\label{sec:base-rate}

So far, our defense involved generating a dummy handshake triplet for every
clearnet connection, while delaying the actual traffic.  This strategy was
successful in preventing circuit fingerprinting attacks, but introduces considerable latency
to the vast majority of Tor traffic.  We can,
however, employ a simple trick to significantly improve the performance without
sacrificing privacy: pad only a \emph{fraction} of connections. For each clearnet connection we
perform a probabilistic choice; with probability $p$ we generate a dummy onion
handshake while delaying the traffic, while with probability $1-p$ we
open the clearnet connection immediately without any padding.

This might appear as sacrificing privacy; the lack of an onion handshake immediately
reveals that the connection is clearnet, while its existence provides evidence that
the connection is onion. However, an accurate judgment requires taking into
account the \emph{base rate} $c$, defined as the percentage of clearnet connections
in the network.
The number of clearnet and
onion connections in Tor are far from equal;
based on experimental measurements from 2018, the network saw about 216 million
exit circuits \cite{n-exit-circs}, while 15 million successful
rendezvous circuits daily \cite{n-rend-circs}, giving an approximate estimate
of $c = 0.93$.

Let $N$ denote the total number of onion handshakes (either dummy or real) and
$S$ denote the connection type which we wish to keep secret, namely $S=0$ for
clearnet and $S=1$ for onion connections. Note that the base rate is equal to $c
= \Pr[S = 0]$ (prior probability).  As already shown, the shape
of an onion handshake can be easily observed by the adversary, but real and
dummy handshakes are indistinguishable. As a consequence, in the following
analysis we assume that $N$ is the only information available to the adversary;
since onion connections always produce a (real) handshake, while clearnet ones
produce a (dummy) handshake with probability $p$, we have that
$Pr[N{=}1 \ |\ S{=}1] = 1$ and $Pr[N{=}1\ |\ S{=}0] = p$.

This leads us to the following result:

\begin{restatable}{theorem}{OptimalAccuracy}\label{thm:optimal-accuracy}
	The accuracy of the optimal classifier predicting
	$S$ from $N$ is equal to
	$ \max \{c, 1 - cp \} $.
\end{restatable}

Note that $c$ is the probability of being correct when guessing blindly.
Essentially, when $N=0$ the adversary has a clear choice (we must have $S=0$).
The case $N=1$ however, is harder: onion connections do produce $N=1$ with higher
probability, but they are far less likely to occur.
If $p \ge \frac{1}{c}-1$, it is in fact more likely that a connection is clearnet,
even if we observe $N=1$.%
\footnote{Note that if we set $p =\frac{1}{c}-1$, the fraction of \emph{padded} connections (wrt all)
will be $cp = 1-c$, exactly equal to the fraction of \emph{onion} connections.}
As a consequence, the optimal adversary always
believes that the connection is clearnet, regardless of $N$, hence he has no advantage
whatsoever over blind guessing.

\paragraph{Tradeoff between privacy and latency/bandwidth}

The tradeoff between privacy and latency is displayed in
\autoref{fig:fractional-leakage-latency}.  Note that the adversary can achieve
an accuracy of $c$ simply by \emph{blind guessing}; but having high accuracy
simply because $c = 0.9$ is not really a problem of the system.  As a
consequence, we follow a standard approach from Quantitative Information Flow
\cite{Smith09} and compute the system's \emph{leakage}, defined as ``accuracy
$-$ prob. of blind guessing''.  We see that for large values of $c$, we can
actually achieve \emph{zero} leakage while padding only a small fraction of
connections.

Regarding the \emph{latency} overhead, our strategy adds a delay to each exit
circuit equal to the time it takes to complete an onion service
handshake. In our experiments (\autoref{sec:handshake-timing}) we found
that onion service handshakes have a mean time of around 3.34 seconds (with a
median time of 2.8 seconds and standard deviation of 1.64 seconds).  Since we
dummy handshake is produced with probability $p$, and only for clearnet connections,
the expected latency
overhead will be $cp \cdot 3.34$ seconds/connection, displayed in
\autoref{fig:fractional-leakage-latency}.  Regarding bandwidth, a single onion
handshake requires 44 cells (22.5 KB) to complete, leading to
an expected overhead of $cp \cdot 22.5$ KB/connection.

As a realistic example, setting $c=0.9$ (slightly lower than the $0.93$ estimate from 2018),
we can completely hide the connection type (zero leakage) by padding only
$p =\frac{1}{c}-1 \approx 11\%$ of traffic, leading to
approximately 334 msecs of latency and 2.25 KB of bandwidth overhead per connection.

\begin{figure}
	\centering
	\includegraphics[width=.9\columnwidth]{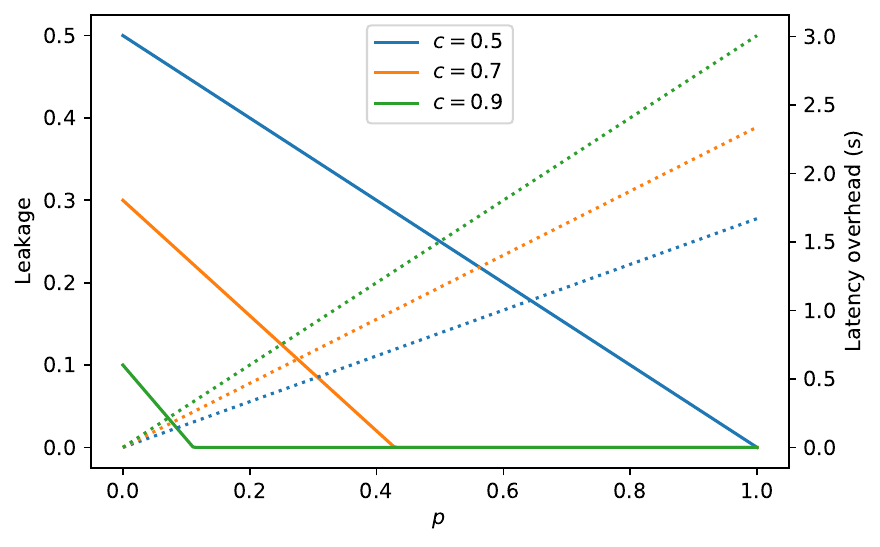}
	\vspace{-10pt}
	\caption{
		Leakage (solid lines) and latency overhead (dotted lines) for the fractional-delay strategy.
	}
	\label{fig:fractional-leakage-latency}
\end{figure}


%% file: 8-low-latency.tex
\section{The zero-delay PCP strategy}\label{sec:pcp}

\begin{figure*}
  \tikzset{>=latex}
  \begin{tikzpicture}
      \tikzset{block/.style= {draw, rectangle, align=center,minimum width=0.6cm,minimum height=0.4cm},
      rblock/.style={fill, shape=rectangle,double,double distance=0.02cm,align=center,minimum width=0.6cm,minimum height=0.4cm},
      }

      \node [block] (preemptive1a) {};
      \node [block, label={preemptive setup}, right=0.0cm of preemptive1a] (preemptive1b) {};
      \node [block, right =0.0cm of preemptive1b] (preemptive1c) {};

      \node [block, right =1.3cm of preemptive1c] (dummy1a) {};
      \node [block, label={dummy handshake}, right=0.0cm of dummy1a] (dummy1b) {};
      \node [block, right =0.0cm of dummy1b] (dummy1c) {};

      \node [block, right =1.3cm of dummy1c] (dummy2a) {};
      \node [block, label={dummy handshake}, right=0.0cm of dummy2a] (dummy2b) {};
      \node [block, right =0.0cm of dummy2b] (dummy2c) {};

      \node [block, right =1.3cm of dummy2c] (dummy3a) {};
      \node [block, label={dummy handshake\phantom{y}}, right=0.0cm of dummy3a] (dummy3b) {};
      \node [block, right =0.0cm of dummy3b] (dummy3c) {};

      \node [block, right =1.2cm of dummy3c] (request) {application-layer clearnet traffic};

      \path[draw,->] 
                  (preemptive1c) edge (dummy1a)
                  (dummy1c) edge (dummy2a)
                  (dummy2c) edge (dummy3a)
                  (dummy3c) edge (request)
                 ; 
  \end{tikzpicture}
  \begin{tikzpicture}
      \tikzset{block/.style= {draw, rectangle, align=center,minimum width=0.6cm,minimum height=0.4cm},
      rblock/.style={fill, shape=rectangle,double,double distance=0.02cm,align=center,minimum width=0.6cm,minimum height=0.4cm},
      }

      \node [block] (preemptive1a) {};
      \node [block, label={preemptive setup}, right=0.0cm of preemptive1a] (preemptive1b) {};
      \node [block, right =0.0cm of preemptive1b] (preemptive1c) {};

      \node [block, right =1.3cm of preemptive1c] (dummy1a) {};
      \node [block, label={dummy handshake}, right=0.0cm of dummy1a] (dummy1b) {};
      \node [block, right =0.0cm of dummy1b] (dummy1c) {};

      \node [block, right =1.3cm of dummy1c] (dummy2a) {};
      \node [block, label={dummy handshake}, right=0.0cm of dummy2a] (dummy2b) {};
      \node [block, right =0.0cm of dummy2b] (dummy2c) {};

      \node [block, right =1.3cm of dummy2c] (dummy3a) {};
      \node [block, label={real handshake\phantom{y}}, right=0.0cm of dummy3a] (dummy3b) {};
      \node [block, right =0.0cm of dummy3b] (dummy3c) {};

      \node [block, right =1.2cm of dummy3c] (request) {application-layer onion traffic};

      \path[draw,->] 
                  (preemptive1c) edge (dummy1a)
                  (dummy1c) edge (dummy2a)
                  (dummy2c) edge (dummy3a)
                  (dummy3c) edge (request)
                  ;
  \end{tikzpicture}
  \caption{Example of padded clearnet and onion connections, containing dummy and real handshake triplets.}
  \label{fig:padded-connections}
\end{figure*}
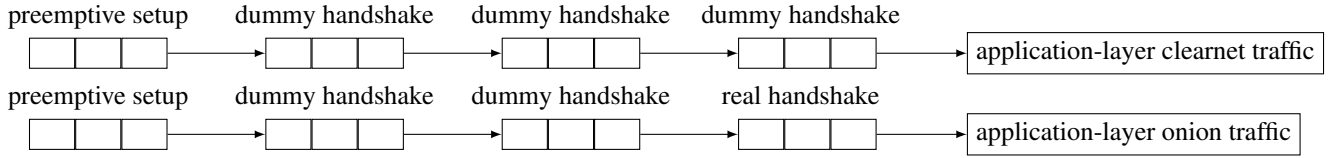

In \autoref{sec:design} we showed that we can construct dummy onion handshakes
that perfectly match both the shape and the timing characteristics of real
ones.  The natural way to use these dummy handshakes is to inject them at the
beginning of a clearnet connection while delaying the actual traffic; this
strategy was shown in \autoref{sec:fractional-delay} to be an effective way of
making clearnet connections indistinguishable from onion ones.

Although the delay can be significantly reduced via \emph{fractional} delay
(\autoref{sec:base-rate}), one might wish to provide a defense without any
delay whatsoever.  If we are not allowed to delay traffic, by the time the
clearnet request arrives it's already \emph{too late} for providing a defense.
But how can we perform dummy handshakes \emph{before} the application data even
arrives? The key insight here is use \emph{preemptive circuits}, a mechanism
that already exists in Tor for performance reasons, but which can be exploited
for defending against fingerprinting. In the following sections we detail a
padding strategy we call \emph{Preemptive Circuit Padding} (PCP).

\paragraph{Preemptive padding}

Creating circuits is a costly operation, due to the use of public-key
cryptography and the latency caused by extending the circuit to three hops.
To optimize performance, the Tor client continuously builds \emph{preemptive
circuits}, that are extended to three hops and then stay dormant until a
proper use for them is found. The preemptive circuit can be later used
as an exit circuit, where the third hop acts as an exit node;
similarly, it can also be used as a rendezvous circuit
where the third hop acts as the rendezvous point.

Tor's preemptive circuit subsystem is not formally specified and depends on the
current implementation. The current Tor client counts the number of available
preemptive circuits every second, and creates new ones until a certain
configurable threshold is reached. The client uses preemptive circuits both for
clearnet and for onion connections.

The preemptive circuit mechanism presents a great opportunity to insert
confusion in the time signature of onion services, since it blends the
\emph{time of creation} with the \emph{time of use} of circuits. Using the
preemptive circuit mechanism, Tor creates circuits before the
actual usage of the circuit has been determined, and hence our machines can
inject padding between the user and the middle node. An adversary observing the
traffic will not be able to infer that this is a preemptive circuit still
waiting to be used, and could confuse it with an active onion circuit that is
completing its onion handshake. Of course, for this defense to be successful,
the padding should be properly constructed, taking also into account the actual
traffic that will pass through the circuit after the preemptive phase.

In the following sections we detail the use of this technique to defend against
circuit fingerprinting. Note that the use of preemptive padding is not limited
to circuit fingerprinting; we envision that it could be useful for building defenses
against guard discovery or traffic confirmation
attacks\cite{traffic-confirmation}.


\paragraph{The zero-delay PCP strategy}

This strategy involves injecting dummy onion handshakes (\autoref{sec:design})
in the preemptive phase. Since each handshake is in fact a \emph{triplet}
(HSDir, Intro and Rendezvous), we employ three preemptive circuits, each having
a distinct role:

\begin{itemize}
  \setlength\itemsep{0.06em}
  \item \emph{HSDir preemptive circuits} carry dummy and real
  HSDir handshakes.
  \item \emph{Introduction preemptive circuits} carry dummy and real
  introduction handshakes.
\item \emph{Exit/Rendezvous preemptive circuits} carry dummy and real
  rendezvous handshakes, as well as (clearnet or onion) application data.
\end{itemize}

Upon creating these three preemptive circuits, we use padding machines
(\autoref{sec:design}) to inject dummy onion handshakes in the corresponding
preemptive circuits.  At this moment, the actual user connection has not
arrived yet, and we know neither when it will arrive, nor whether it
will be a clearnet or onion connection.  As a consequence, we keep
\emph{repeating} triplets of dummy handshakes in the preemptive circuits, using
a \emph{random delay} between them.  When the user connection arrives, we
proceed depending on its type.
\begin{itemize}
  \setlength\itemsep{0.06em}
\item Clearnet connections: the HSDir and introduction preemptive circuits
  are terminated, and the Exit/Rendezvous preemptive circuit is used as an exit
  circuit.
\item Onion connections: the preemptive circuits are used as HSDir,
  introduction and rendezvous circuits respectively, carrying the real handshakes
  and the application data.
\end{itemize}

The resulting traffic pattern is displayed in \autoref{fig:padded-connections}.
Both clearnet and onion connections contain dummy onion handshake triplets as
well as application-layer data, while onion connections also contain the real
onion handshake.  Since dummy and real handshakes have identical shape an
adversary cannot distinguish the two cases.

Note that the number of repeated dummy handshakes and the
\emph{delay between them} are crucial parameters for this defense to be effective.
Let $\lambda_u$ be the expected number of connections requested by the user per unit of
time (rate), which depends on the user behavior.  We propose to use a random
delay between dummy handshakes sampled from an \emph{exponential distribution}\footnote{%
   We could alternatively use some heavy-tailed distribution (eg Pareto).
}
with rate $\lambda_d = \varphi\cdot\lambda_u$, where $\varphi$ is a parameter of the
defense.  Intuitively, $\lambda_d$ expresses the mean number of dummy handshakes
per unit of time, while $\varphi$ expresses the mean number of dummy handshakes
\emph{per user connection}.  Larger values of $\varphi$ offer better privacy at
the cost of extra bandwidth.

We can imagine several ways to implement this defense in practice. For instance,
one could use the fractional-delay defense from \autoref{sec:fractional-delay}, until the system's usage increases
and low-latency is required. Then Tor can switch to PCP using
the $\lambda_u$ most recently observed, which can be then dynamically adjusted
on regular intervals. If the system becomes idle, we can switch back
to fractional-delay to reduce bandwidth overhead.

The analytic and experimental evaluations discussed in the remaining of this
section show that, under assumptions on the user's thinking time, properly
choosing $\varphi$ can offer significant defense against circuit fingerprinting
at a reasonable bandwidth cost.

\subsection{Analytic evaluation}\label{sec:evaluation-pcp}

We start with an analytic evaluation, obtaining an expression for the
accuracy of an optimal adversary which is useful for tuning the parameters of
the defense.

In the example of \autoref{fig:padded-connections},
the adversary can observe that three handshakes took place,
since each one has a distinguishable traffic pattern (\autoref{fig:dummy-patterns}),
but cannot tell whether we have three dummy handshakes, or two dummy and a real one.
Although the two connection instances appear
indistinguishable to the adversary, PCP
will only probabilistically produce such identical instances, so the
\emph{probability} to produce each one of them can be used to 
infer the type of connection.

Following the notation of \autoref{sec:base-rate},
let $D$ denote the number of \emph{dummy} handshakes sent on the
circuits (whiled $N$ denotes the \emph{total} number of handshakes, either dummy or
real). The adversary cannot observe $D$ directly (dummy and real handshakes
look identical), but can potentially observe $N$. In the case of
a clearnet connection we have $N = D$, while for onion connection we have
$N = D+1$ (because of the additional real handshake), so after
observing $N$ we can use the distribution of $D$ to infer which type of
connection is more likely. In the example of \autoref{fig:padded-connections},
the adversary observes $N=3$ but does not know whether we have a clearnet
connection with $D=3$ or an onion one with $D=2$. So, the probabilities
$\Pr[D=2]$ and $\Pr[D=3]$ need to \emph{similar}; if we perform too few dummy
handshakes, and $D=3$ has negligible probability, the adversary can infer that we
have an onion connection.

The first step in our analytic evaluation is to study the distribution of $D$.
Since dummy
handshakes are repeated until the user connection arrives, $D$ crucially depends
on the user's \emph{think time}, that is the time between consecutive handshakes
for a new connection. As our user model, we assume that the think time follows
an \emph{exponential} distribution with rate $\lambda_u$;
this is a commonly used model, although
not always accurate \cite{PaxsonF95}.  Recall that the
time between dummy handshakes is
also sampled from an exponential distribution with rate $\lambda_d$; the
question here is to compute how many dummy handshakes will be injected while
waiting for the user connection to arrive.

Intuitively, this number $D$ depends on the relationship between $\lambda_u$ and 
$\lambda_d$. If $\lambda_d$ is much larger it means that we are generating
dummies more quickly than the user is opening connections, so $D$ will be large;
similarly, if $\lambda_u$ is much larger then $D$ will be small.
In fact, we show that the distribution of $D$ solely depends on their
ratio
$\varphi = \nicefrac{\lambda_d}{\lambda_u}$, which brings us to the following
result (all proofs are in \autoref{sec:proofs}). 

\begin{restatable}{theorem}{DummyGeometric}\label{thm:dummy-geometric}
	The number $D$ of injected dummy triplets
	follows a geometric distribution with
	parameter $\nicefrac{1}{1+\varphi}$, that is
	$$\Pr[D = k] = p(1-p)^{k} \quad \text{for }p = \textstyle\frac{1}{1+\varphi},k \ge 0 ~.$$
\end{restatable}

\begin{figure}
	\centering
	\includegraphics[width=.9\columnwidth]{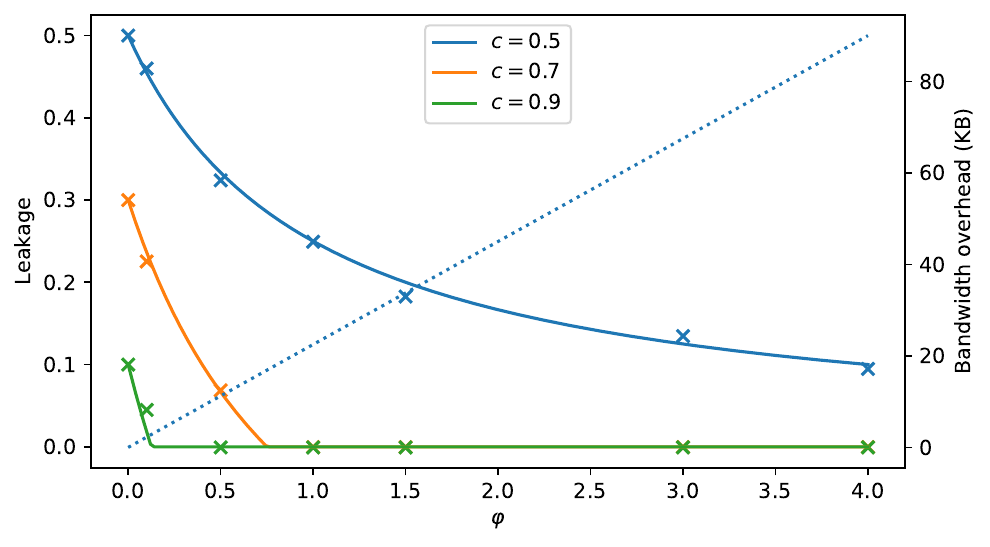}
	\vspace{-10pt}
	\caption{Leakage (solid lines) and bandwidth overhead (dotted lines) for the
      PCP strategy.
      Experimental results are also shown as individual crossed points.
   }
	\label{fig:experiment5_results}
\end{figure}

\paragraph{Base rate and classification accuracy}

As discussed in \autoref{sec:base-rate}, the asymmetry in the base rate
makes it harder for the adversary to produce
confident guesses and classifications of onion circuits.
In the example of \autoref{fig:padded-connections},
even if dummy handshakes are relatively sparse and $\Pr[D = 3]$ is much
smaller than $\Pr[D=2]$, it would be a \emph{base rate fallacy}
to immediately conclude that we have an onion connection \cite{base-rate}.
If $c = \Pr[S = 0]$ is close to $1$, it could still be more likely that a clearnet
connection happened to produce unusually many dummy handshakes.

We are finally ready to state our
main result, showing that the accuracy of an optimal adversary who predicts $S$
from $N$ depends only on $\varphi$ and $c$.

\begin{restatable}{theorem}{OptimalAccuracy}\label{thm:optimal-accuracy}
	The accuracy of the optimal classifier predicting
	$S$ from $N$ is equal to
	$\max\{c, 1 - c \frac{\varphi}{\varphi+1}\}$.
\end{restatable}

As discussed in \autoref{sec:base-rate}, the adversary can achieve an accuracy
of $c$ by simply blind guessing. In
\autoref{fig:experiment5_results} we plot the system's \emph{leakage}, instead
of the accuracy, which expresses \emph{how much worse} the attack becomes due to
the system's output.

There are two main conclusions from these results: first, increasing $\varphi$
decreases the system's leakage, which is expected since it leads to more dummy
handshakes.  Second, the base rate $c$ works in our advantage: higher values
require stronger evidence to support an accurate detection, meaning that we can
achieve the same leakage with a smaller overhead $\varphi$.

These expressions are particularly useful for tuning the defense's parameters.
For instance, assuming $c = 0.7$ (much smaller than the $c=0.93$ estimate from
2018), we see that setting $\varphi=1$ is enough to make the system have
\emph{zero leakage}, meaning that the adversary is no better than blind
guessing.  Fixing $\varphi$ in turn allows us to configure the parameter
$\lambda_d$, the one needed for implementing the defense, by estimating
$\lambda_u$ from the user's previous behavior and setting
$\lambda_d = \varphi\cdot\lambda_u$.

\paragraph{Avoiding $\lambda_u$ leaks}

Although $\lambda_d$ is only known to the client, the middle node observes the
times between dummy handshakes (i.e. \emph{samples} from the distribution) and
hence it can estimate $\lambda_d$, and thus $\lambda_u$, leading to potential
fingerprinting issues. However, the middle node changes for every circuit and
on average it only learns $\varphi$ samples, so for small $\varphi$
(e.g. $\varphi = 0.5$) the estimate will be very rough and this type of leakage
could be tolerated. If we want to further protect $\lambda_u$ there are ways to
do so: first, we can increase $\lambda_d$ by adding sufficient noise to it so that its
initial value is obfuscated.

A different approach for bigger values of $\varphi$ would be to keep multiple
preemptive circuits open for each type (with different middle nodes), and
distribute the dummy handshakes randomly among them (in the extreme case send a
single handshake per circuit) while sending the application traffic on the
circuit with the most recent dummy handshake.  Although this approach has worse
performance because of the additional preemptive circuits, the security
analysis of the defense is not affected. The advantages are twofold: first, we
hide the $\lambda_u$ from individual middle nodes, and second, it works around
the restrictions imposed by Tor's RELAY\_EARLY mechanism.

\paragraph{Tradeoff between privacy and bandwidth}\label{sec:pcp-bandwidth}

Although our PCP scheme does not increase latency at all, it does incur a
bandwidth overhead due to the dummy handshakes.  Since we inject an average of
$\varphi$ handshakes per connection, the expected overhead will be
$\varphi\cdot 22.5$ KB/connection.  Hence, as expected, there is a tradeoff
between privacy and bandwidth, displayed in \autoref{fig:experiment5_results}.
Note that the overhead does not depend on $\lambda_d$ or $\lambda_u$, but only
on their ratio $\varphi$; moreover, in contrast to the fractional-delay
strategy (\autoref{sec:fractional-delay}), the overhead does not depend on $c$,
since we pad both clearnet and onion traffic.

\paragraph{Timing analysis}

The PCP defense has potential timing
fingerprints that are not present in our simpler delay-based defense.
In particular, even though the timing patterns
of onion handshakes can be imitated accurately (\autoref{sec:design}),
the timing between the application-layer traffic and the
padding could be used as a timing fingerprint.
We believe that a relative large value of $\lambda_d$ will alleviate this issue
(note that the low-latency of PCP is useful in busy periods when $\lambda_u,\lambda_d$ are
large).
Still, we consider that
analyzing and evaluating more advanced defenses against such timing
fingerprints to be beyond the scope of this paper and left as future work.

Note that the fractional-delay strategy from \autoref{sec:fractional-delay}
does not suffer from such potential fingerprints, since the application-layer
traffic immediately follows the handshake. We view the fractional-delay as
trading latency for robustness with regards to timing fingerprints.

\subsection{Experimental evaluation}\label{sec:pcp-evaluation}

To confirm our analytic evaluation, we performed a final experiment using the
same single-site scenario as in \autoref{sec:evaluation-delay}.  Implementing
the defense in our deployed framework of \autoref{sec:wtf} would require
deploying further framework modifications to the network (e.g. finer control of
preemptive circuits, and working around Tor's restricting RELAY\_EARLY
mechanism) so we chose to simulate the PCP mechanism similarly to the
experiments of \autoref{sec:evaluation-delay}.

We used various values of $\varphi$ and $c$, and performed the following
procedure: For each onion connection, we injected dummy handshakes in its three
circuits as follows: we first sampled the user thinking time $t$ from an
exponential distribution with a fixed rate $\lambda_u = 4$ reqs/hour.\footnote{%
  This choice is arbitrary, since $\lambda_u$ depends on the user behavior
  for which it is hard to obtain data. Still,
  since we keep $\varphi$ fixed (setting $\lambda_d = \varphi\cdot\lambda_u$),
  $\lambda_u$ only affects the bandwidth overhead and not the
  obtained privacy.
}
Then, we
continuously sampled the time of each dummy handshake from an exponential
distribution with rate $\lambda_d = \varphi\cdot\lambda_u$,
injecting a triple of dummy handshakes in each of the three circuits (HSDir,
introduction and rendezvous), until reaching the time $t$.  The same procedure
was followed to pad exit circuits, except that for each one of them two fake
ones (fake HSDir and fake Intro) were created and added to the dataset.

The results for exit and rendezvous circuits (of a single website) are shown in
\autoref{fig:experiment5_results}, showing the accuracy for each combination of
$\varphi,c$ are a data point over the analytic expressions. We see that the
accuracy is identical to that predicted from \autoref{thm:optimal-accuracy}, up
to a small experimental error. This is in fact quite remarkable, since the
classifier takes $N$ as its input, while in the experiments $N$ is not directly
observable, but the classifier receives the direct cell sequences. The
correspondence of the results means that the classifier in fact manages to
learn $N$ (and use it to infer the connection type) but nothing more, since the
traffic patterns are otherwise similar.

For other types of circuits we obtained the same results, omitted due to
space restrictions. We also repeated the whole set of experiments with
$\lambda_u = 8$ reqs/hour, obtaining the same results, hence verifying the analytic
conclusion
that privacy depends on $\varphi$
alone, and not on the individual values $\lambda_d,\lambda_u$.


%% file: 9-discussion.tex
\section{Discussion}\label{sec:discussion}

We briefly discuss several new directions in both improving the attacks
and exploring avenues for more effective defenses.

\paragraph{Defining circuit fingerprinting}

We believe that formalizing the circuit fingerprinting problem further can make
defenses and evaluations more robust. We believe that a formal security game
similar to IND-CPA where the adversary is asked to distinguish between onion
and exit cell traces seems like a fruitful direction forward.

Furthermore, instead of doing classification on a per-circuit basis, it could
prove more effective to do the classification on the global flow of
connections. This way we could use the timing correlations between circuits of
the onion service protocol as a feature for classification during experiments.

It is worth noting that previous works on website and circuit fingerprinting
did classification on a per-circuit level because investigating timing
correlations on a global level presents design issues and significantly
increases the uncertainty of the classifier. Our goal with this work has been
to push the circuit fingerprinting attacker into requiring access to precise
timing information and not being able to do fingerprinting with simpler
features like cell sequences.

\paragraph{Adversary capabilities}

As discussed in \autoref{sec:model}, it is common practice in the website
fingerprinting literature to consider the network adversary being able to
derive the precise cell contents of a circuit given raw TCP traces. However
because of the importance of the precise cell sequence in the circuit
fingerprinting world, we have a strong interest in weakening the network
adversary.

We believe we can achieve this using Tor's adaptive padding framework by
increasing the number of simultaneous padding circuits to a point where a
network adversary will find it hard to distinguish which circuit sends which
cell. Combined with Tor's move to a single guard per client \cite{one-guard} it
makes the work of the network adversary much harder since there will always be
simultaneous traffic to the single guard.

\paragraph{Future work}

While we implemented the initial padding machines into Tor, we did not
implement our PCP defense and opted for a simulation due to the development
effort required and the need for multiple iterations of the padding
techniques. We believe that implementing PCP in upstream Tor should be a next
step in the research against circuit fingerprinting.

Furthermore, our padding framework has room for improvement, since it currently
lacks features like \emph{flow control}, \emph{load balancing} and
\emph{probabilistic state transitions}\cite{tor-circ-padding-docs}.

Finally, we believe that the advanced padding techniques introduced in this
work can also be useful to hide the service-side onion service traffic from
circuit fingerprinting adversaries. The main issue is that onion services can
be distinguished by the great asymmetry between their incoming and outgoing
traffic, which is not present in regular Internet connections. To emulate that
volume of outgoing traffic using padding machines we would need realistic user
traffic patterns that we could replay with padding machines.


%% file: 10-conclusion.tex
\section{Conclusion}

In this work we introduced novel padding-based techniques that can be used to
protect Tor onion service clients against circuit fingerprinting attacks. In
the process of doing so, we demonstrated new fingerprints that can be used to
distinguish onion service circuits. We also built and deployed a versatile
adaptive padding framework based on WTF-PAD that can be used to implement such
padding-based defenses.

Our work shows that adaptive padding defenses that utilize preemptive circuits
can be fruitful to thwart circuit fingerprinting adversaries with no additional
latency cost.


%% file: appendix.tex
%
\newenvironment{Reason}{\begin{tabbing}\hspace{2em}\= \hspace{1cm} \= \kill}
{\end{tabbing}\vspace{-1em}}
\newcommand\Step[2] {#1 \> $\begin{array}[t]{@{}llll}\displaystyle #2\end{array}$ \\}
\newcommand\StepR[3] {#1 \> $\begin{array}[t]{@{}llll}\displaystyle #3\end{array}$
\` {\RF \makebox[0pt][r]{\begin{tabular}[t]{r}``#2''\end{tabular}}} \\}
\newcommand\WideStepR[3] {#1 \>
$\begin{array}[t]{@{}ll}~\\\displaystyle #3\end{array}$ \`
{\RF \makebox[0pt][r]{\begin{tabular}[t]{r}``#2''\end{tabular}}} \\}
\newcommand\Space {~ \\}
\newcommand\RF {\small}

\section{Dataset sizes}\label{sec:dataset-sizes}

In this section we present the number of circuits in each dataset for each
experiment. These circuits describe the multi site scenarios of our tests.  The
number of circuits are shown below:

\paragraph{Delay-based strategy experiment}
\begin{center}
	\begin{tabular}{ |c|c| }
	\hline
	\textbf{HSDir} & \textbf{Fake HSDir} \\
	\hline
	15536 & 15536 \\
	\hline\hline
	\textbf{Introduction} & \textbf{Fake Introduction} \\
	\hline
	15508 & 15508 \\
	\hline\hline
	\textbf{Rendezvous} & \textbf{Padded-exit} \\
	\hline
	14955 & 15008 \\
	\hline
	\end{tabular}
\end{center}

\paragraph{Delay-based strategy over localhost experiment}

\begin{center}
	\begin{tabular}{ |c|c| }
	\hline
	\textbf{HSDir} & \textbf{Fake-HSDir} \\
	\hline
	15000 & 15000 \\
    \hline\hline
	\textbf{Introduction} & \textbf{Fake Introduction} \\
	\hline
	15000 & 15000 \\
	\hline\hline
	\textbf{Rendezvous} & \textbf{Padded-exit} \\
	\hline
	15000 & 15000 \\
	\hline
	\end{tabular}
\end{center}

\paragraph{PCP zero-latency strategy experiment}

\begin{center}
	\begin{tabular}{ |c|c| }
	\hline
	\textbf{HSDir} & \textbf{Fake-HSDir} \\
	\hline
	15000 & 15000 \\
    \hline\hline
	\textbf{Introduction} & \textbf{Fake Introduction} \\
	\hline
	15000 & 15000 \\
	\hline\hline
	\textbf{Rendezvous} & \textbf{Padded-Exit} \\
	\hline
	15000 & 15000 \\
	\hline
	\end{tabular}
\end{center}

\section{Onion service cell sequence diagrams}\label{sec:cell-diagrams}

In this section we present cell sequence diagrams for all onion
circuits. \autoref{fig:hsdir-circuit-cell-flow} shows the cell flow of HSDir
circuits, and \autoref{fig:rendezvous-circuit-cell-flow} shows the cell flow of
rendezvous circuits.

\begin{figure}[h]
  \tikzset{>=latex}
  \begin{tikzpicture}
	\draw (0,0) -- (0,4);
	\draw (2,0) -- (2,4);
	\draw (4,0) -- (4,4);
	\draw (6,0) -- (6,4);
	\draw (8,0) -- (8,4);

	\draw [dotted, arrows={-triangle 45}]
	(0, 4) --  node [midway, fill=white, text centered]
	{ \scriptsize EXTEND } (4, 3.75);

	\draw [dotted, arrows={triangle 45-}]
	(0, 3.25) --  node [midway, fill=white, text centered]
	{ \scriptsize EXTENDED } (4, 3.5);

	\draw [dotted, arrows={-triangle 45}]
	(0, 3.0) -- node [midway, fill=white, text centered]
	{ \scriptsize EXTEND } (6, 2.75);

	\draw [dotted, arrows={-triangle 45}]
	(6, 2.5) -- node [midway, fill=white]
	{ \scriptsize EXTENDED } (0, 2.25);

    \draw[dashed, arrows={-triangle 45}] (0, 2.0) -- node [midway, fill=white, text centered] {\scriptsize EXTEND} (8, 1.75);
    \draw (0, 2.0) -- (3.4, 1.89);

    \draw[dashed, arrows={-triangle 45}] (8, 1.5) -- node [midway, fill=white, text centered] {\scriptsize EXTENDED} (0, 1.25);
    \draw (3.3, 1.35) -- (0, 1.25);

    \draw[dashed, arrows={-triangle 45}] (0, 1.0) -- node [midway, fill=white, text centered] {\scriptsize BEGIN\_DIR + DATA} (8, 0.75);
    \draw (0, 1) -- (2.8, 0.91);

    \draw[dashed, arrows={-triangle 45}] (8, 0.5) -- node [midway, fill=white, text centered] {\scriptsize CONNECTED + DATA + END} (0, 0.25);
    \draw (2.3, 0.32) -- (0, 0.25);

	\node at (0,4.25) {\scriptsize Client};
	\node at (2,4.25) {\scriptsize Guard};
	\node at (4,4.25) {\scriptsize Middle 1};
	\node at (6,4.25) {\scriptsize Middle 2};
	\node at (8,4.25) {\scriptsize HSDir};
  \end{tikzpicture}

  \caption{Cell sequence of an HSDir request}
  \label{fig:hsdir-circuit-cell-flow}
\end{figure}
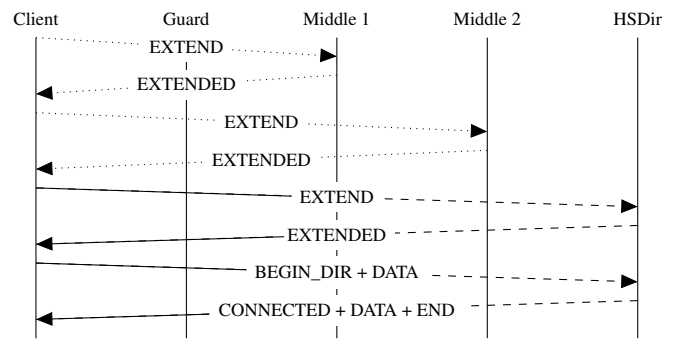

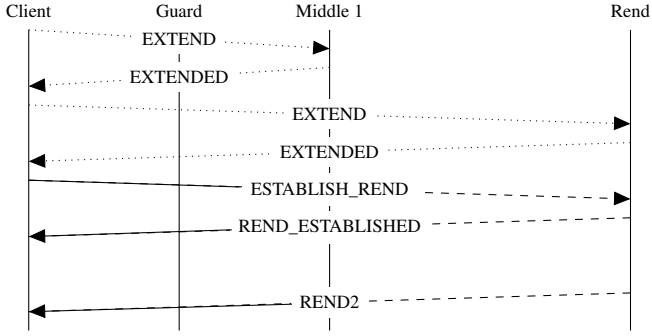
\begin{figure}[h]
  \tikzset{>=latex}
  \begin{tikzpicture}
	\draw (0,0) -- (0,4);
	\draw (2,0) -- (2,4);
	\draw (4,0) -- (4,4);
	\draw (8,0) -- (8,4);

	\draw [dotted, arrows={-triangle 45}]
	(0, 4) --  node [midway, fill=white, text centered]
	{ \scriptsize EXTEND } (4, 3.75);

	\draw [dotted, arrows={triangle 45-}]
	(0, 3.25) --  node [midway, fill=white, text centered]
	{ \scriptsize EXTENDED } (4, 3.5);

	\draw [dotted, arrows={-triangle 45}]
	(0, 3.0) -- node [midway, fill=white, text centered]
	{ \scriptsize EXTEND } (8, 2.75);

	\draw [dotted, arrows={-triangle 45}]
	(8, 2.5) -- node [midway, fill=white]
	{ \scriptsize EXTENDED } (0, 2.25);

    \draw[dashed, arrows={-triangle 45}] (0, 2.0) -- node [midway, fill=white, text centered] {\scriptsize ESTABLISH\_REND} (8, 1.75);
    \draw (0, 2.0) -- (2.8, 1.91);

    \draw[dashed, arrows={-triangle 45}] (8, 1.5) -- node [midway, fill=white, text centered] {\scriptsize REND\_ESTABLISHED} (0, 1.25);
    \draw (2.7, 1.34) -- (0, 1.25);


    \draw[dashed, arrows={-triangle 45}] (8, 0.5) -- node [midway, fill=white, text centered] {\scriptsize REND2} (0, 0.25);
    \draw (3.5, 0.35) -- (0, 0.25);

	\node at (0,4.25) {\scriptsize Client};
	\node at (2,4.25) {\scriptsize Guard};
	\node at (4,4.25) {\scriptsize Middle 1};
	\node at (8,4.25) {\scriptsize Rend};
  \end{tikzpicture}

  \caption{Cell sequence of a rendezvous handshake}
  \label{fig:rendezvous-circuit-cell-flow}
\end{figure}

\section{Padding Machines}\label{sec:padding-machines}

In this section we provide the source code and state machine of padding
machines for creating dummy introduction handshakes.

\subsection{Client-side introduction machine}

The source code below specifies a client-side padding machine that creates
dummy introduction handshakes. The state diagram of the padding machine can be
seen in \autoref{fig:client-intro}.

\begin{figure}[h]
\begin{tikzpicture}[shorten >=1pt,node distance=3.5cm,on grid,auto]
  \tikzstyle{every state}=[fill={rgb:black,1;white,10}]

    \node[state,initial,scale=0.7]   (extend)                    {EXTEND};
    \node[state,scale=0.7 ] (introduce)  [right of=extend]    {INTRODUCE};
    \node[state,accepting,scale=0.7]           (end)  [right of=introduce]    {END};

    \path[->]
    (extend) edge [bend left]  node[scale=0.7]{}    (introduce)
    (introduce) edge [bend left]  node {}    (end);

\end{tikzpicture}
\caption{Client-side introduction handshake padding machine}
\label{fig:client-intro}
\end{figure}
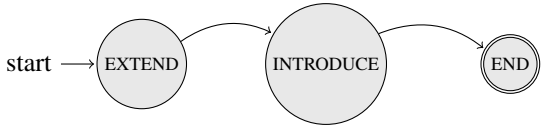

\begin{lstlisting}[language=C,basicstyle=\tiny]
circpad_machine_spec_t *client_machine = tor_malloc_zero(sizeof(circpad_machine_spec_t));
client_machine->name = "client_intro_circ";
client_machine->conditions.apply_state_mask = CIRCPAD_CIRC_OPENED;
client_machine->target_hopnum = 2;

/* This is a client machine */
client_machine->is_origin_side = 1;

/* Start the machine when the preemptive phase ends */
client_machine->conditions.apply_purpose_mask =
  circpad_circ_purpose_to_mask(CIRCUIT_PURPOSE_C_PREEMPTIVE);

/* Keep the machine around if it is in the CIRCUIT_PADDING purpose (but do not
 * try to take over other machines in that purpose). */
client_machine->conditions.keep_purpose_mask =
    circpad_circ_purpose_to_mask(CIRCUIT_PURPOSE_C_CIRCUIT_PADDING);

/* Two states: EXTEND, INTRODUCE (and END) */
circpad_machine_states_init(client_machine, 2);

/* We transition from EXTEND to the next step (INTRODUCE) after receiving
 * PADDING_NEGOTIATED (which replaces EXTENDED). */
client_machine->states[0].next_state[CIRCPAD_EVENT_NONPADDING_RECV] = 1;
/* We transition from INTRODUCE state to END, when we receive a a padding cell
 * (INTRODUCE_ACK) from the other side */
client_machine->states[1].next_state[CIRCPAD_EVENT_PADDING_RECV] = CIRCPAD_STATE_END;

/* During INTRODUCE phase we need to send a single padding cell (INTRODUCE1) */
client_machine->states[0].length_dist.type = CIRCPAD_DIST_UNIFORM;
client_machine->states[0].length_dist.param1 = 1;
client_machine->states[0].length_dist.param2 = 1;
client_machine->states[0].max_length = 1;
/* Clients send their padding as soon as possible */
client_machine->states[0].iat_dist.type = CIRCPAD_DIST_UNIFORM;
client_machine->states[0].iat_dist.param1 = 0;
client_machine->states[0].iat_dist.param2 = 0;

/* Register the machine */
client_machine->machine_num = smartlist_len(machines_sl);
circpad_register_padding_machine(client_machine, machines_sl);
\end{lstlisting}

\subsection{Relay-side introduction machine}

The source code below specifies a relay-side padding machine that creates
dummy introduction handshakes. The state diagram of the padding machine can be
seen in \autoref{fig:relay-intro}.

\begin{figure}[h]
\begin{tikzpicture}[shorten >=1pt,node distance=3.5cm,on grid,auto]
  \tikzstyle{every state}=[fill={rgb:black,1;white,10}]

    \node[state,initial,scale=0.7]   (extended)                    {EXTENDED};
    \node[state,scale=0.7] (introduce_ack)  [right of=extended]    {INTRODUCE\_ACK};
    \node[state,accepting,scale=0.7]           (end)  [right of=introduce_ack]    {END};

    \path[->]
    (extended) edge [bend left]  node[scale=0.7]{}    (introduce_ack)
    (introduce_ack) edge [bend left]  node {}    (end);

\end{tikzpicture}
\caption{Relay-side introduction handshake padding machine}
\label{fig:relay-intro}
\end{figure}
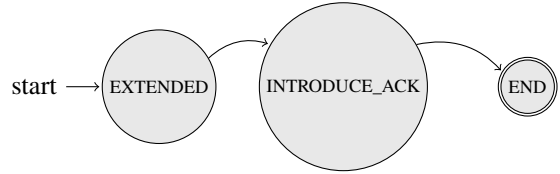

\begin{lstlisting}[language=C,basicstyle=\tiny]
circpad_machine_spec_t *relay_machine = tor_malloc_zero(sizeof(circpad_machine_spec_t));
relay_machine->name = "relay_intro_circ";
relay_machine->conditions.apply_state_mask = CIRCPAD_CIRC_OPENED;
relay_machine->target_hopnum = 2;

/* This is a relay machine */
relay_machine->is_origin_side = 0;

/* Two states: EXTENDED, INTRODUCE_ACK (and END) */
circpad_machine_states_init(relay_machine, 3);

/* For the relay-side machine, we want to transition from EXTENDED to INTRODUCE_ACK
 * upon receiving the first padding cell by the relay-side machine (INTRODUCE1) */
relay_machine->states[0].next_state[CIRCPAD_EVENT_PADDING_RECV] = 1;
/* We transition from INTRODUCE_ACK state to END, when we send the first
 * padding cell (INTRODUCE_ACK) */
relay_machine->states[1].next_state[CIRCPAD_EVENT_PADDING_SENT] =
  CIRCPAD_STATE_END;

/* During EXTENDED phase we need to send a single padding cell (EXTENDED) */
relay_machine->states[0].length_dist.type = CIRCPAD_DIST_UNIFORM;
relay_machine->states[0].length_dist.param1 = 1;
relay_machine->states[0].length_dist.param2 = 1;
relay_machine->states[0].max_length = 1;
/* We use the LogLogistic(alpha=3.19, beta=0.14) distribution that simulates the
 * round-trip delay to the introduction point */
relay_machine->states[0].iat_dist.type = CIRCPAD_DIST_LOG_LOGISTIC;
/* param1 is alpha, param2 is 1.0/beta */
relay_machine->states[0].iat_dist.param1 = 3.19;
relay_machine->states[0].iat_dist.param2 = 0.313;

/* During INTRODUCE_ACK phase we need to send a single padding cell
 * (INTRODUCE_ACK) using the same delay distribution as above) */
relay_machine->states[1].length_dist.type = CIRCPAD_DIST_UNIFORM;
relay_machine->states[1].length_dist.param1 = 1;
relay_machine->states[1].length_dist.param2 = 1;
relay_machine->states[1].max_length = 1;
/* We use the LogLogistic(alpha=3.19, beta=0.14) distribution that simulates the
 * round-trip delay to the introduction point */
relay_machine->states[1].iat_dist.type = CIRCPAD_DIST_LOG_LOGISTIC;
/* param1 is alpha, param2 is 1.0/beta */
relay_machine->states[1].iat_dist.param1 = 3.19;
relay_machine->states[1].iat_dist.param2 = 0.313;

/* Register the machine */
relay_machine->machine_num = smartlist_len(machines_sl);
circpad_register_padding_machine(relay_machine, machines_sl);

\end{lstlisting}

\section{Delay-based strategy experiment}

Similarly with the metrics of \autoref{Tab:experiment4_statistics}, in this
section we present metrics for the delay-based strategy experiment on this
section.  The metrics are TPR and FPR, alongside with Precision for the
interested circuit in each case.  Combined with accuracy results on
\autoref{Tab:experiment3_accuracy}, they describe the classification
performance for both Decision Tree and SVM classifiers for our multi site
scenario.

\begin{table*}[t]
	\caption{True Positive and False Positive Rates in comparison with Precision (delay-based strategy) (A positive result means that the classifier thought that a circuit is a real onion circuit)}
	\centering
	\begin{tabular}{|c|c|c|c|c|c|c|c|}
	\hline
	\multirow{2}{*}{} & & \multicolumn{2}{c|}{\textbf{Fake-HSDir vs HSDir}} & \multicolumn{2}{c|}{\textbf{Fake-Intro vs Intro}} & \multicolumn{2}{c|}{\textbf{Padded-Exit vs Rend}} \\
	\cline{3-8}
	& & \begin{tabular}{@{}c@{}}Dec. Tree\end{tabular} & \begin{tabular}{@{}c@{}}SVM\end{tabular} & \begin{tabular}{@{}c@{}}Dec. Tree\end{tabular} & \begin{tabular}{@{}c@{}}SVM\end{tabular} & \begin{tabular}{@{}c@{}}Dec. Tree\end{tabular} & \begin{tabular}{@{}c@{}}SVM\end{tabular} \\
	\hline
	\multirow{3}{*}{Multi-Close} & TPR & 0.03 & 0.03 & 0 & 0 & 0.91 & 0.98 \\
	\cline{2-8}
	& FPR & 0 & 0 & 0 & 0 & 0.31 & 0.72 \\
	\cline{2-8}
	& Precision & 1.00 & 1.00 & 1.00 & - & 0.73 & 0.56 \\
	\hline
	\multirow{3}{*}{Multi-Open} & TPR & 0.04 & 0.04 & 0 & 0 & 1.00 & 0.49 \\
	\cline{2-8}
	& FPR & 0 & 0 & 0 & 0 & 0.99 & 0.50 \\
	\cline{2-8}
	& Precision & 1.00 & 1.00 & - & - & 0.49 & 0.49 \\
	\hline
	\end{tabular}
	\label{Tab:experiment3_statistics}
\end{table*}

\begin{table*}[t]
	\caption{True Positive and False Positive Rates in comparison with Precision (delay-based strategy over localhost) (A positive result means that the classifier thought that a circuit is a real onion circuit)}
	\centering
	\begin{tabular}{|c|c|c|c|c|c|c|c|}
	\hline
	\multirow{2}{*}{} & & \multicolumn{2}{c|}{\textbf{Fake-HSDir vs HSDir}} & \multicolumn{2}{c|}{\textbf{Fake-Intro vs Intro}} & \multicolumn{2}{c|}{\textbf{Padded-Exit vs Rend}} \\
	\cline{3-8}
	& & \begin{tabular}{@{}c@{}}Dec. Tree\end{tabular} & \begin{tabular}{@{}c@{}}SVM\end{tabular} & \begin{tabular}{@{}c@{}}Dec. Tree\end{tabular} & \begin{tabular}{@{}c@{}}SVM\end{tabular} & \begin{tabular}{@{}c@{}}Dec. Tree\end{tabular} & \begin{tabular}{@{}c@{}}SVM\end{tabular} \\
	\hline
	\multirow{3}{*}{Multi-Close} & TPR & 0 & 0 & 1.00 & 1.00 & 0.61 & 0.86 \\
	\cline{2-8}
	& FPR & 0 & 0 & 1.00 & 1.00 & 0.64 & 0.88 \\
	\cline{2-8}
	& Precision & - & - & 0.49 & 0.49 & 0.47 & 0.48 \\
	\hline
	\multirow{3}{*}{Multi-Open} & TPR & 0 & 0 & 1.00 & 1.00 & 1.00 & 1.00 \\
	\cline{2-8}
	& FPR & 0 & 0 & 1.00 & 1.00 & 1.00 & 1.00 \\
	\cline{2-8}
	& Precision & - & - & 0.50 & 0.50 & 0.50 & 0.50 \\
	\hline
	\end{tabular}
	\label{Tab:experiment4_statistics}
\end{table*}

\section{Proofs of \autoref{sec:evaluation-pcp}}\label{sec:proofs}

\DummyGeometric*
\begin{proof}
	Let $T$ be a random variable modelling the time when the
	user's request arrives. $T$ is assumed to follow an exponential
	distribution with rate $\lambda_u$, with pdf:
	\begin{equation}\label{eq1}
		f_{\,T}(t) \Wide= \lambda_u e^{-\lambda_u t}
		~.
	\end{equation}

	Now fix some value $T=t$, and consider the number of dummies
	generated before time $t$.
	Since the inter-arrival time between dummies follows an exponential
	distribution with rate $\lambda_d = \varphi\lambda_u$, it is well known that
	the number of dummies generated before time $t$
	follows a Poisson distribution with parameter $\lambda_d t$,
	in other words
	\begin{equation}\label{eq2}
		\Pr[D = k \ |\ T=t ] \Wide= e^{-\lambda_d t} \frac{(\lambda_d t)^k}{k!}
		~.
	\end{equation}
	Letting $p = \frac{1}{1+\varphi}$, note that
	\begin{equation}\label{eq3}
		p \Wide= \frac{\lambda_u}{\lambda_d + \lambda_u}
		\quad\text{and}\quad
		1-p \Wide= \frac{\lambda_d}{\lambda_d + \lambda_u}
		~.
	\end{equation}
	Finally, we have that:
	\begin{Reason}
	\Step{}{
		\Pr[D = k]
	}
	\Step{$=$}{
		\int_{0}^{\infty} \Pr[D = k \ |\  T=t]\cdot f_{\,T}(t) \;dt
	}
	\StepR{$=$}{\eqref{eq1},\eqref{eq2}}{
		\int_{0}^{\infty}
		e^{-\lambda_d t} \frac{(\lambda_d t)^k}{k!}
		\lambda_u e^{-\lambda_u t}
	}
	\Step{$=$}{
		\frac{ {\lambda_d}^k \lambda_u } { k! }
		\int_{t=0}^{\infty}
		t^k e^{-(\lambda_d + \lambda_u)t }
	}
	\StepR{$=$}{$\int_0^{\infty} x^n e^{-ax}dx = \frac{n!}{a^{n+1}}$}{
		\frac{ {\lambda_d}^k \;\lambda_u } {(\lambda_d + \lambda_u)^{k+1}  }
	}
	\StepR{$=$}{\eqref{eq3}}{
		p (1-p)^k
		~.
	}
	\end{Reason}

\end{proof}

\OptimalAccuracy*
\begin{proof}
	The Bayes classifier (which is well known to be optimal) maps an observation
	$k$ to the label $s$ with
	the highest \emph{posterior} probability $\Pr [ S=s \ |\  N = k ]$.
	Note that this is equivalent to selecting the label $s$ with the highest
	\emph{joint} probability $\Pr [ S=s,  N = k ]$.
	This classification is correct for all pairs $(s,k)$ such that $s$ is
	the one with the highest $\Pr [ S=s,  N = k ]$, for that $k$.
	Hence the classifier's accuracy is equal to the probability of such a pair appearing,
	which is given by
	\begin{equation}\label{eq4}
		\text{Accuracy} \Wide=
		\sum_{k=0}^{\infty} \max_{s\in\{0,1\}} \Pr [ S = s, N = k ]
		~.
	\end{equation}

	To compute the joint probabilities, note that when $S =0$ (exit circuit),
	the adversary only observes the generated dummy triplets, that is $N =
	D$. On the other hand, when $S = 1$ (onion circuit), the adversary
	observes the dummy triplets plus the real one, that is $N = D+1$. Since
	$D$ follows a geometric distribution with parameter $p = \frac{1}{1+\varphi}$
	(\autoref{thm:dummy-geometric}), and
	(by definition) $c = \Pr[S = 0] = 1 - \Pr[S = 1]$, we get the following
	expressions for the joint probabilities:
	\begin{align*}
		\Pr [ S = 0, N = k ]
			&\Wide=
			c\cdot p(1-p)^k
			\\
		\Pr [ S = 1, N = k]
			&\Wide=
			\begin{cases}
				0 & \text{if } k = 0 \\
				(1-c)p(1-p)^{k-1} & \text{if } k > 0
			\end{cases}
	\end{align*}

	We now need to find which label $s \in \{0,1\}$ gives
	the max for each observation $k$ in \eqref{eq4}.
	For $k=0$, the max is clearly given by $s=0$
	since $\Pr[S = 1, N = 0] = 0$
	(intuitively,
	onion circuits produce at least one triplet).
	For $k>0$, on the other hand, the max is given by $s=0$
	iff
	\begin{align}
		\Pr [ S = 0, N = k ]
			&\Wide\ge
			\Pr [ S = 1, N = k ]
			&\Leftrightarrow \nonumber\\
		c\cdot p(1-p)^k
			&\Wide\ge
			(1-c)p(1-p)^{k-1}
			&\Leftrightarrow \nonumber\\
		c
			&\Wide\ge
			1-c(1-p) \label{eq5}
	\end{align}
	Note that \eqref{eq5} does not depend on $k$, in other
	words either $s=0$ or $s=1$ give the max for all $k > 0$
	simultaneously. We consider the two cases.

	Case $c \ge 1-c(1-p) $: in this case the max in \eqref{eq4}
	is given by	$s=0$
	for all $k \ge 0$, hence the accuracy becomes
	$$
		\text{Accuracy}
		\Wide=
		\sum_{k=0}^{\infty} \Pr [ S = 0, N = k ]
		\Wide=
		\Pr [ S = 0 ] \Wide= c
		~.
	$$

	Case $c < 1-c(1-p) $: in this case the max in \eqref{eq4}
	is given by	$s=0$
	for $k=0$ and by $s=1$ for $k > 0$, hence the accuracy becomes
	\begin{Reason}
		\Step{}{
			\text{Accuracy}
		}
		\Step{$=$}{
			\Pr [ S = 0, N = 0 ] +
			\sum_{k=1}^{\infty} \Pr [ S = 1, N = k ]
		}
		\Step{$=$}{
			\Pr [ S = 0, N = 0 ] +
			\big(\Pr[S = 1] - \Pr [ S = 1, N = 0 ]\big)
		}
		\Step{$=$}{
			cp + (1-p)
		}
		\Step{$=$}{
			1 - c(1-p)
			~.
		}
	\end{Reason}

	Finally, combining the two cases we get:
	\begin{Reason}
		\Step{}{
			\text{Accuracy}
		}
		\Step{$=$}{
			\begin{cases}
				c & \text{if } c \ge 1-c(1-p) \\
				1 - c (1-p) & \text{otherwise}
			\end{cases}
		}
		\Step{$=$}{
			\max\{ c , 1-c(1-p) \}
		}
		\StepR{$=$}{$p = \frac{1}{1+\varphi}$}{
			\max\{ c , 1-c\frac{\varphi}{1+\varphi}) \}
			~.
		}
	\end{Reason}

\end{proof}

\begin{figure}
	\centering
	\includegraphics[width=\columnwidth]{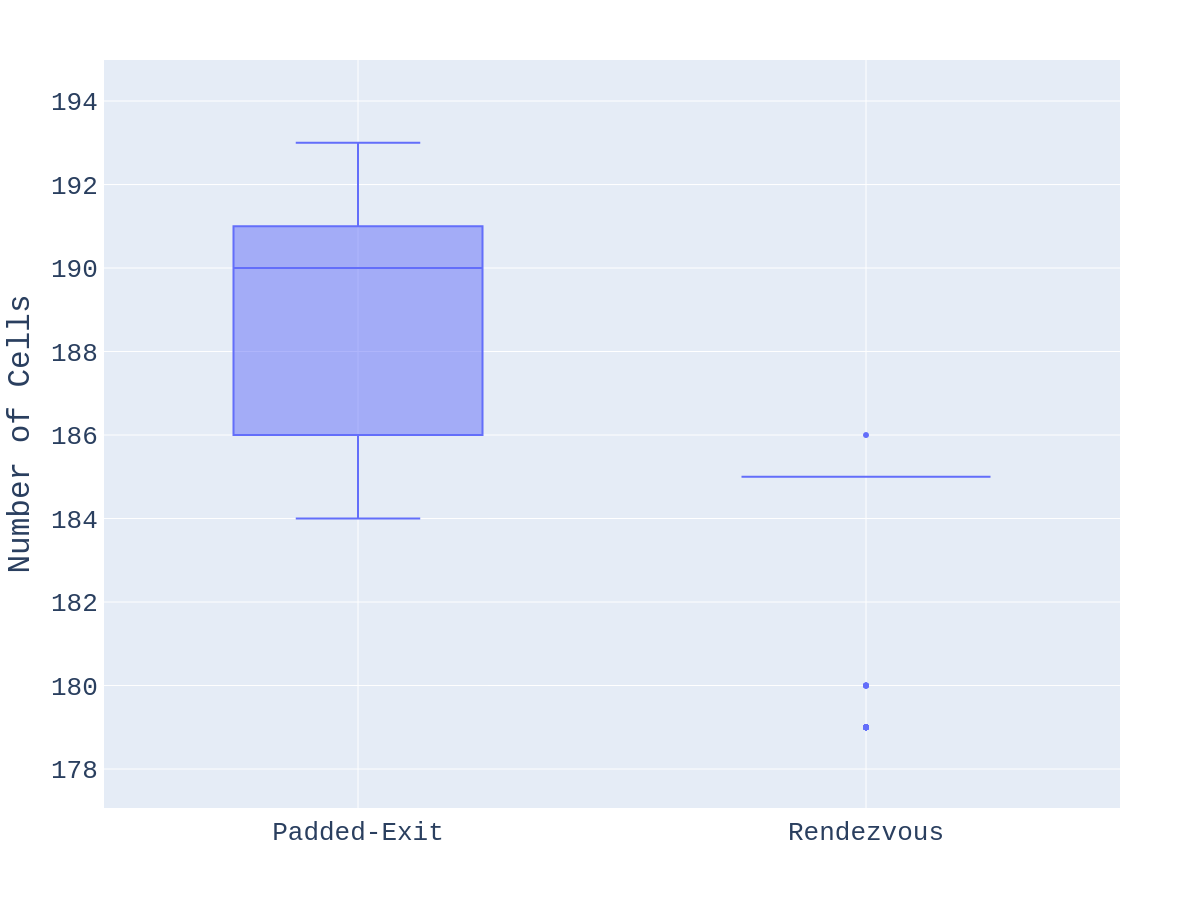}
	\caption{Cell packing difference between rendezvous and exit circuits}
	\label{fig:cell_dif_exp3}
\end{figure}
